\documentclass[12pt]{amsart}
\usepackage{amsmath,amsfonts,amssymb} \usepackage{color}
\usepackage[all,cmtip]{xy}
\usepackage{graphicx}

 \newcommand{\bQ}{{\mathbb Q}}

 \newcommand{\bP}{\mathbb{P}} \newcommand{\cM}{\mathcal{M}} \newcommand{\cP}{\mathcal{P}}
 
  \newcommand{\cD}{\mathcal{D}} 
  \newcommand{\cL}{\mathcal{L}}
  
 \newcommand{\bZ}{\mathbb{Z}}
 \newcommand{\bC}{\mathbb{C}}
 \newcommand{\pd}{\partial}
\newcommand{\Mbar}{\overline{\mathcal M}}
\newcommand{\bx}{\mathbf{x}} \newcommand{\by}{\mathbf{y}}
\newcommand{\vac}{|0\rangle}  \newcommand{\lvac}{\langle 0|}

    \DeclareMathOperator{\Tr}{Tr}
  
\DeclareMathOperator{\Span}{span}

\newcommand{\be}{\begin{equation}}
\newcommand{\ee}{\end{equation}}
\newcommand{\bea}{\begin{eqnarray}}
\newcommand{\eea}{\end{eqnarray}}
\newcommand{\ben}{\begin{eqnarray*}}
\newcommand{\een}{\end{eqnarray*}}

\newcommand{\half}{\frac{1}{2}}

\newtheorem{cor}{Corollary}[section]

 \newtheorem{prop}[cor]{Proposition}
 \newtheorem{thm}[cor]{Theorem}
\theoremstyle{remark}

 \newtheorem{rmk}[cor]{Remark}

\definecolor{A}{rgb}{.75,1,.75}

\definecolor{yellow}{rgb}{1,1,0}
\definecolor{orange}{rgb}{1,.7,0}
\definecolor{red}{rgb}{1,0,0}
\definecolor{white}{rgb}{1,1,1}

 \makeindex
\begin{document}
\title
{Grothendieck's Dessins d'Enfants in a Web of Dualities}

\author{Jian Zhou}
\address{Department of Mathematical Sciences\\Tsinghua University\\Beijing, 100084, China}
\email{jianzhou@mail.tsinghua.edu.cn}

\begin{abstract}
In this paper we show that counting Grothendieck's dessins d'enfants
is universal in the sense that some other enumerative problems are either special cases 
or directly related to it. 
Such results provide concrete examples that support a proposal made in the paper to study 
various dualities from the point of view of group actions on the moduli space of theories.
Connections to differential equations of hypergeometric type 
can be made transparent from this approach,
suggesting a connection to mirror symmetry.
\end{abstract}

\maketitle

\section{Introduction}

In physics literature,
a duality often means the equivalence of
two different theories.
Because a physical theory is given by a specific model,
a duality then means the equivalence of two different models.
When there are several different theories which
are  dual to each other,
then we are in the situation of having a web of dualities.

From a physical point of view a duality between
two different theories is not only useful
but also natural.
This is because if two different theories
are proposed to describe the same physical world, 
two different theories must be equivalent.
On the contrary,
from a mathematical point of view a duality
is often very surprising because in examples
of dualities proposed by physicists,
the two theories  that are dual to each other
are often constructed using quite different
mathematical objects that do not seem to be
directly related to each other.
Hence most mathematical works have focused on mathematically verifying
such dualities or their consequences.
Now there have accumulated many examples of dualities
in the literature.
The abundance of the dualities
then suggest to understand them in a unified way.
In this paper we will make some attempts
in this directions.
Instead of trying to provide geometric constructions to make any individual
duality mathematically more transparent,
we propose to study all the dualities of different theories
all together.
One can imagine to work with a space of all possible theories.
A duality between two theories can then be interpreted
as there is an action of an element in some group
that transforms one theory to the other,
and vice versa.

Such an approach follows an old tradition in mathematics.
In mathematics it is a common practice to put
some objects with the same properties together
to form a set.
This set is often called the space of such objects.
One then studies the geometric structures of this space.
by looking for the symmetries of this space,
which are described by the action of a group on it.
The existence of a web of dualities then strongly
suggests that there should be a transitive group action
on the space of theories,
and dualities are just the examples of symmetries of this space.
Then we are in a similar situation
as in Klein's Erlangen Program,
but we are working with an infinite-dimensional space
of theories and the transformation group
between different theories is also infinite-dimensional,
but the action is still expected to be transitive.
Now from our point of view,
saying two different theories are dual to each other might be an
understatement.
A complete picture is that both theories lie in an infinite-dimensional moduli space
of theories,
and this space is homogeneous under
suitable group actions.
For a related approach at the level of abstract quantum field theory,
see \cite{Wang-Zhou}.

One of the simplest and most well-studied
infinite-dimensional space
with a homogenous action of
an infinite-dimensional group
is Sato's semi-infinite Grassmannian
that arises in the theory of KP hierarchy.
Since it is well-known that the space of KP tau-functions
is an infinite-dimensional homogeneous space,
our basic strategy is to embed any given theory into the KP hierarchy,
whenever that is possible.

The KP hierarchy is a system of infinitely many nonlinear partial differential equations
for a sequence of infinitely functions in a sequence of infinitely many variables.
A solution can be encoded in a single function called the KP tau-function
of the solution,
and the KP system
then becomes a system of Hiroto bilinear relations.
By the approach of the Kyoto school,
one can apply the boson-fermion correspondence to convert the system
of Hirota relation as the Pl\"ucker relations
on the Sato Grassmannian,
hence one is led to a setting of infinite-dimensional algebraic geometry.
Fortunately,
it is not very difficult to work with  this infinite-dimensional setting.
For our purpose,
we only need to consider the big cell of the Sato Grassmannian,
this means to specify a KP tau-function is to specify an element in the big cell of Sato Grassmannian,
which are given by a sequence of series of the form:
\be
z^n + \sum_{m \geq 1} A_{m,n} z^{-m}.
\ee 
The coefficients $\{A_{m,n}\}$ are called the {\em affine coordinates} of the tau-function.
The unimaginable power of the theory of KP hierarchy is that it reduces 
the study of a system of infinitely many nonlinear partial differential equations
to essentially a problem in infinite-dimensional linear algebra:
The determination of the matrix $(A_{mn})$.
This problem, unsurprisingly,
can be often solved with the help of representation theory of groups.

One of the first physical theories studied from the point of view of KP hierarchy is
the theory of 2D topological gravity.
Originally the connection with the theory of integrable hierarchy was made 
in the context of matrix models.
Witten \cite{Witten1} related the theory of 2D topological gravity
to intersection   numbers of $\psi$-classes on moduli spaces $\Mbar_{g,n}$ of algebraic curves.
Their generating series is a tau-function of the KdV hierarchy satisfying the Virasoro constraints  \cite{DVV}
by the famous Witten Conjecture/Kontsevich Theorem \cite{Witten1, Kon}.
In Kyoto school's approach \cite{MJD},
there are three pictures to describe the tau-function of the KdV hierarchy:
(a) as an element of the Sato infinite-dimensioanl Grassmannian;
(b) as a vector in the bosonic Fock space, i.e. the space of symmetric functions;
(c) as a vector in the fermionic Fock space, i.e., the semi-infinite wedge space.
Hence one can describe the Witten-Kontsevich tau-function in each of these pictures.
For $\tau_{WK}$
it was given by Kac-Schwarz \cite{Kac-Schwarz} in the first picture,
by Alexandrov \cite{Alexandrov}  in the second picture:
\be
Z_{WK} = e^{\hat{W}} 1= 1 + \sum_{n \geq 1} \frac{1}{n!} \hat{W}^n 1,
\ee
where the operator $\hat{W}$ is a cut-and-join type differential operator defined by
\ben
\hat{W}
& = &\frac{2}{3} \sum_{k,l \geq 0}(k+\frac{1}{2})u_k \cdot (l+\frac{1}{2}) u_l \cdot \frac{\pd}{\pd u_{k+l-1}} \\
& + & \frac{\lambda^2}{12} \sum_{k, l \geq 0} (k+l+\frac{5}{2}) u_{k+l+2} \frac{\pd^2}{\pd u_k\pd u_l}
+ \frac{1}{\lambda^2} \frac{u_0^3}{3!} + \frac{u_1}{16},
\een
and by the author \cite{Zhou-WK} in the third picture:
\be \label{eqn:WK}
Z_{WK} = e^A \vac,  \quad \quad \quad
A =  \sum_{m,n \geq 0} A_{m,n} \psi_{-m-\frac{1}{2}}\psi_{-n-\frac{1}{2}}^*,
\ee
where $A_{m,n}$ are some numbers related to the asymptotic expansions of
Airy functions $Ai(x)$ and $Bi(x)$ whose explicit expressions can be found in {\em loc. cit}.

In \cite{Zhou-WK} Virasoro constraints alone are used to derive \eqref{eqn:WK}.
Balogh and Yang \cite{BY} make a different derivation based on the KdV hierarchy.
This latter approach has been generalized to KP hierarchy,
and furthermore,
a formula for the bosonic $n$-point functions based on the fermionic formula has been
derived by the author in \cite{Zhou-Emerg}.
We apply these results to the partition functions of Hermitian one-matrix models
in \cite{Zhou-MM} since they are known to be tau-functions of the KP hierarchy.
We will do the same for $Z_{Dessins}$ in this paper.
This is possible because due to the results of Zograf \cite{Zog} and Kazarian-Zograf \cite{Kaz-Zog},
$Z_{Dessins}$ is a family of tau-functions of the KP hierarchy,
satisfying some Virasoro constraints and has a cut-and-join type representation in the bosonic picture.
We fill in the missing piece in the fermionic picture in this paper.

Now we come to the second subject in the title: Grothendieck's dessins d'enfants.
The dessins, together with the mdouli spaces
$\cM_{g,n}$ and $\Mbar_{g,n}$,
are the central geometric objects in Grothendieck's
1600 page manuscript  ``La Longue Marche a travers la Theorie de Galois" (Long March through Galois theory)
written in 1981,
and also his {\em Esquisse d'un Programme} (Sketch of a Programme).
Of course the arithmetic applications that involves the absolute Galois group $Gal(\overline{\bQ}/\bQ)$
are the main purpose
of Grothendieck's program.
Nevertheless,
it is very easy to notice many similarities
between the arithmetic aspects of the theory of dessins,
some of which will be recalled or reported in this paper or other parts of the series.
We hope some kind of unification can be eventually achieved
between the arithmetical and the enumerative aspects of
the theory of dessins.

This paper is the first of a series of papers devoted to
 some relations between the enumerations
of Grothendieck's dessins d'enfants and enumerative problems of the moduli spaces $\cM_{g,n}$
and their Deligne-Mumford compactifications $\Mbar_{g,n}$.
Such relations are parallel to the relationship between the Hurwitz numbers
that enumerate the branched coverings of Riemann spheres and intersection numbers on
$\Mbar_{g,n}$
and its generalizations that the author is also involved with.
The similarities and difference between dessin counting and Hurwitz numbers
make the comparison between them extremely rich.
Both are related to well-studied theories in mathematical physics
and can be studied by methods well developed in mathematical physics.
For example,
the Hurwitz numbers are closely related to local mirror symmetry of toric Calabi-Yau
$3$-folds,
and hence the idea of dualities between two different theory play a dominant role,
while the dessins are more closely related to the 2D topological gravity
where integrable hierarchies play a dominant role.
One of the purpose of our series of papers is to unify them in the same frameworks
in the hope that insights for one of them can shed some lights on the other.
In particular,
in this approach it is natural to study the symmetries of integrable hierarchies from the point of view of dualities,
and it is also natural to study dualities from the point of view of the symmetries of integrable hierarchies.
One can also use other frameworks, for example,
theory of Frobenius manifolds, or the theory of Eynard-Orantin topological recursions.
In this paper,
we will focus on the framework of KP hierarchy.
In other parts of the series we will also consider these other frameworks.

Our results in the planned series of papers can be roughly summarized in the following
flow chart:
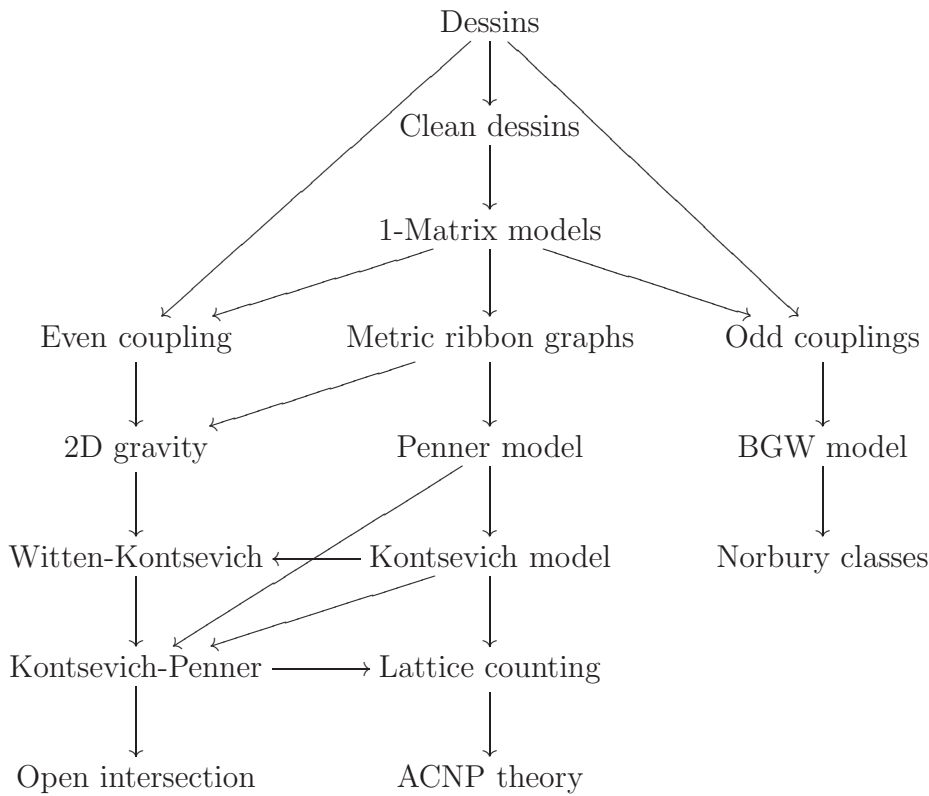
\ben
\xymatrix{
& \text{Dessins}  \ar[d] \ar[lddd] \ar[rddd] & \\
& \text{Clean dessins}  \ar[d] & \\
& \text{1-Matrix models} \ar[ld] \ar[d] \ar[rd] &  \\
\text{Even coupling} \ar[d] & \text{Metric ribbon graphs} \ar[d] \ar[ld]
& \text{Odd couplings} \ar[d] \\
\text{2D gravity} \ar[d] & \text{Penner model} \ar[d] \ar[ldd]  & \text{BGW model} \ar[d] \\
\text{Witten-Kontsevich} \ar[d]  & \text{Kontsevich model} \ar[d] \ar[l] \ar[ld] & \text{Norbury classes} \\
\text{Kontsevich-Penner}  \ar[d] \ar[r]
& \text{Lattice counting}  \ar[d] & \\
\text{Open intersection} & \text{ACNP theory} &
}
\een
Here by ACNP theory we mean the theory in \cite{ACNP}.
We should emphasize that every item in this flow chart involves some way of counting dessins.
Because this paper is already  very long,
we will treat only the cases that involves counting dessins and clean dessins.
The discussions for those theories that involve clean dessins with metrics (metric ribbon graphs)
will be presented in Part II of the series.

There is a similar flow chart starting with the Hurwitz numbers to be presented in Part II of the series.
In both of these flow charts,
when one follows the arrows,
in the end one always reaches at a geometric theory.
This will explained from the point of Eynard-Orantin topological recursions in a later part of the series.

Many theories in string theory developed
in recent years are already in these two pictures.
Some other theories can also fit in these pictures,
for example,
Witten's $r$-spin intersection numbers.
We hope to consider them later.

In this paper we will focus on relating the theories in these pictures
from the point of view that each of them gives a tau-function of the KP hierarchy.
Because the tau-functions of the KP hierarchy lie on the Sato Grassmannian,
they are naturally related to each other by the action of Kac-Moody group
$\widehat{GL}(\infty)$.
For each of them we will discuss the following topics:
1. Virasoro constraints; 2. Cut-and-join representations;
 3. Kac-Schwarz operators;
4. Quantum spectral curves;
5. Fermionic representation;
6. Formula for $n$-point functions.
We can also consider the following topics:
7. Eynard-Orantin topological recursions, and
8. Gromov-Witten type theory.
Because of limit of space we will consider them in later parts of the series.

One of the hallmarks in the study of mirror symmetry is the appearance of hypergeometric
equations.
For a long time it is missing in the study of 2D gravity.
So the works centered around mirror symmetry and the works centered around Witten Conjecture/Kontsecich Theorem
do not seem to match with each other.
With the advent of Eynard-Orantin topological recursion \cite{EO} and
quantum spectral curves \cite{GS},
a recent framework called the emergent geometry of KP hierarchy have been proposed
to provide a unified picture.
A type of results from this point of view
obtained in this paper is that the quantum spectral curves
are differential equations of hypergeometric type.

It is well-known that
any KP tau-function can be expanded in the Schur polynomials
\be
\tau = \sum_{\mu \in \cP}
c_\mu s_\mu,
\ee
where $s_\mu$ are Schur functions, $\cP$ denotes
the set of all possible partitions (including the empty partition).
There are not very many results for the explicit explicit expressions for $c_\mu$,
some examples in the literature
include formulas for the topological vertex and
topological string theory on toric Calabi-Yau 3-folds \cite{ADKMV, Deng-Zhou1, Deng-Zhou2},  for the Witten-Kontsevich tau-function
$\tau_{W}$ in \cite{Zhou-WK, Zhou-Emerg},
and   for the fat graphs via Hermitian one-matrix models in  \cite{Zhou-MM}.
The idea of using the fermionic picture comes from \cite{ADKMV}
and can be traced back to \cite{DMP}.
In this paper we will prove the following formula for the generating series enumerating
Grothendieck's {\em dessins d'enfants}:
\be
Z_{Dessins} = \sum_{\mu \in \cP} s_\mu \cdot \prod_{e \in \mu} \frac{s \cdot (u+c(e))(v+c(e))}{h(e)},
\ee
where in the product
$e$ runs through the boxes of Young diagram of  the partition $\mu$,
and $c(e)$ and $h(e)$ denotes the content and hook length of $e$ respectively.
Based on this we then get a formula for the $n$-point functions
associated with the dessin counting.

The proofs of these formulas are based on the following earlier results in the literature:
(1) The partition function $Z_{Dessins}$ defined as
a generating series of the weighted count of dessins
is a tau-function of the KP hierarchy \cite{Zog}. (2)  $Z_{Dessins}$ satisfies
some Virasoro constraints \cite{Kaz-Zog}.
(3) A fermionic formula for a tau-function of the KP hierarchy
based on boson-fermion correspondence  in \cite{Zhou-Emerg}.
(4) A formula proved in \cite{Zhou-Emerg} for the $n$-point functions associated with a tau-function
 of the KP hierarchy based on the fermionic formula in (3).

We also study the counting of clean dessins in the same fashion.
First we show that the counting function of clean dessins considered in \cite{Kaz-Zog} 
is equivalent to the partition function of Hermitian one-matrix model,
hence the results in an earlier paper of the author \cite{Zhou-MM} can be applied.
Then we can compare the affine coordinates of the dessin tau-function
and that of the clean dessin tau-function (cf. Theorem \ref{thm:A-A}).

Next we study the counting of clean dessins with only even vertices.
By our result this is equivalent to Hermitian one-matrix model with only even coupling constants.
The relevant integrable hierarchy to its partition function was observed 
by Witten \cite{Witten1} to be the Volterra lattice equation, also called 
the discrete KdV equation. 
The problem of finding Virasoro constraints for this theory was solved recently 
in \cite{DLYZ}. 
It turns out that one should introduce a modified partition in this case using difference operator. 
We will show that the Virasoro constraints for the modified partition function 
of Hermitian one-matrix model with even couplings
leads to a cut-and-join representation of this partition function.
By comparing with the cut-and-join representation for the dessin tau-function,
we find that the modified partition function is
a special case of dessin tau-function.
This establishes a very surprising connection between counting dessins in general 
with counting clean dessins with only even vertices.

We next turn to matrix models with only odd coupling constants.
We find that the cut-and-join representation of 
generalized BGW model in \cite{Alex-BGW}
is also a special case of the cut-and-join representation of the dessin tau-function.

The natural next step, which we leave to  Part II of the series, 
is to consider metric ribbon graphs.

The rest of this paper is arranged as follows.
In Section 2 we review some preliminary backgrounds dessins and KP hierarchy.
In Section 3 we identify the element in the Grassmannian that corresponds to the enumeration
of dessins.
In the next three Sections we relate dessin counting to other counting problems. 
In Section 4 we treat the case of counting of clean dessins,
in Section 5 the case of counting clean dessins with even vertices,
and in Section 6 the case of generalized BGW theory.
In the final Section 7 we make some concluding remarks.
In the Appendix we present the proofs of our main theorems for counting dessins. 

\section{Preliminaries}

In this Section we briefly recall some definitions and facts
about dessins (see e.g. \cite{Schneps})
and  KP hierarchy (see e.g. \cite{MJD}).

\subsection{Belyi's Theorem and dessins}

Let us first recall Belyi's Theorem.
Let $X$ be an algebraic curve defined over $\bC$.
Then $X$ is defined over $\bar{\bQ}$
if and only if there exists a holomorphic function
$f : X \to \bC\bP^1$
such that $f$ only ramifies over $0$, $1$, $\infty$.
Such a morphism is called a {\em Belyi morphism},
and the pair $(X, f)$ is called a {\em Belyi pair}.
Given a Belyi pair,
the inverse image $f^{-1}([0, 1]) \subset X$
of the real line segment $[0, 1] \subset \bC\bP^1$
 is a connected bicolored graph $\Gamma$,
whose vertices are the preimages of $0$ marked in black
and $1$ marked in white respectively.
The complement $X - \Gamma$ is a cell.
Such a bicolor graph is called a {\em Grothendieck's dessin
d'enfants}.
There is a one-to-one correspondence between the
isomorphism classes of Belyi pairs and connected bicolored ribbon graphs.
This is called the Grothendieck correspondence.

\subsection{Dessins vs. clean dessins}
A Belyi morphism is said to be a {\em pre-clean} Belyi morphism
if all the ramification orders over $1$
are less than or equal to $2$, and {\em clean}
if they are all exactly equal to $2$.
Clean graphs correspond to ribbon graphs (also called fat graphs).
See \cite{Schneps} for the following observation:
If $\alpha: X \to \bC\bP^1$ is a Belyi morphism,
then $\beta = 4\alpha(1 - \alpha)$ is a clean
Belyi morphism.

Since the work of t'Hooft \cite{tHt}, counting fat graphs (i.e. clean dessins)
have found many applications in string theory.
We will show in this paper that counting ordinary dessins should be equally
important,
and it might be even more fundamental than counting clean dessins.

\subsection{KP hierarchy }

One direct way to describe the tau-functions of the KP hierarchy is
to use Sato's Grassmannian.
Let $z^{1/2}\bC[[z,z^{-1}]]$ be the space of formal series in $z$ of half integral powers.
An element $V$ in the big cell $Gr_0$ of Sato's Grassmannian is specified by a sequence of series
\be
\Psi_n(z) = z^{n+1/2} + \sum_{m=0}^\infty A_{m,n} z^{-m-1/2}, \;\; n \geq 0,
\ee
where the coefficients $A_{m,n}$ are called the {\em affine coordinates} of $V$.
This sequence is called the {\em normalized basis} of $V$.
Other approaches have also been used in the literature.
For example,
any sequence of series of the form
\be
\Phi_i(z) = z^{i-1} (1 + \sum_{j=1}^\infty a_{j,i}z^{-j}),
\;\;\; i \geq 1
\ee
is a basis of $\bC[[z,z^{-1}]]$.
Given such a basis,
one can obtain from it the normalized basis:
\ben
&& \Psi_0(z) = z^{1/2}\Phi_1(z), \\
&& \Psi_1(z) = z^{1/2}\Phi_2(z) - a_{1,2} z^{1/2}\Phi_1(z), \\
&& \Psi_2(z) = z^{1/2}\Phi_3(z) - a_{1,3} z^{1/2}\Phi_2(z) - (a_{2,3} - a_{1,3}a_{1,2}) z^{1/2}\Phi_1(z), \\
&& \cdots \cdots \cdots \cdots
\een

Given the  normalized basis $\{\Psi_n\}_{n \geq 0}$,
one can add to it $\{z^{-n+1/2} \}_{n \geq 1}$ to get a basis
of $z^{1/2}\bC[[z,z^{-1}]]$.
Therefore,
if we have another element $\tilde{V}$ of the Sato Grassmannian,
specified by a normalized basis $\{\tilde{\Psi}_n\}_{n \geq 0}$,
then one can define a linear map on $z^{1/2}\bC[[z,z^{-1}]]$
as follows:
\ben
&& \Psi_n(z) \mapsto \tilde{\Psi}_n(z),
\;\;\; n \geq 0, \\
&& z^{-n+1/2} \mapsto z^{-n+1/2}, \;\;\; n \geq 1.
\een
Under this map $V$ is mapped to $\tilde{V}$.

Given an element $V$ with normalized basis $\{\Psi_n\}_{n \geq 0}$,
the tau-function $Z_V$ can be obtained as follows.
One first gets the Pl\"ucker coordinates of $V$ as follows:
\be
Z_V: = \Psi_0(z) \wedge \Psi_1(z) \wedge \cdots
= \sum_\mu \det(A_{m_i,n_j})_{1\leq i , j \leq k } \cdot
|\mu\rangle,
\ee
where the summation is taken over all partitions
$\mu$, expressed as
$$(m_1, \cdots , m_k | n_1, \cdots , n_k)$$
in Frobenius notation.
Here
\be
|\mu\rangle = (-1)^{n_1 + n_2 + \cdots + n_k}
\prod_{i=1}^k \psi_{-m_i-\frac{1}{2}} \psi_{-n_i-\frac{1}{2}}^*|0\rangle_F
\ee
in the fermionic Fock space,
$\vac_F$ is the fermionic vacuum:
\be
\vac_F = z^{\frac{1}{2}} \wedge z^{\frac{3}{2}} \wedge \cdots,
\ee
$\psi_{r}$ is the operator $z^{r} \wedge$,
and $\psi_{r}^*$ is the adjoint operator of $\psi_{-r}$.
The bosonic Fock space can be  taken to be  the space $\Lambda$ of symmetric functions.
One can take the Schur functions the Schur functions $\{ s_\nu\}_{\mu\in \cP}$
to be an additive basis of $\Lambda$,
where $\cP$ denotes the set of all partitions,
including the empty partition.
One can also take the Newton functions
$$p_\nu = p_{\nu_1} \cdots p_{\nu_l},$$
to form an additive basis.
These two bases are related to each other by the Frobenius formula \cite{Macdonald}:
\be \label{eqn:Newton-Schur}
p_\nu = \sum_\nu \chi^\mu_\nu s_\mu,
\ee
where $\chi^\mu_\nu$ are character values of the symmetric groups.
To get the bosonic version of the tau-function,
one can use an isomorphism called the boson-fermion correspondence.
Explicitly, this isomorphism is given by
\cite[Theorem 9.4]{MJD}:
\begin{equation}\label{eqn:boson-fermion}
s_\mu = \lvac e^{\sum_{n=1}^\infty \frac{p_n}{n}\alpha_n} |\mu\rangle ,
\end{equation}
where
\begin{equation*}
\alpha_n = \sum_{r\in \mathbb{Z} + \frac{1}{2}}:\psi_{-r}\psi^*_{r+n}:
\end{equation*}

For more detailed discussions the reader can also see \cite{Zhou-Emerg}
where one can also find the following result:

\begin{thm} \label{thm:n-point}
The $n$-point function associated with $Z_V$ is given by the following formula:
\be
G^{(n)}(\xi_1, \ldots, \xi_n)
= (-1)^{n-1} \sum_{\text{$n$-cycles}}
\prod_{i=1}^n \hat{A}(\xi_{\sigma(i)}, \xi_{\sigma(i+1)})
-  \frac{\delta_{n,2}}{(\xi_1-\xi_2)^2},
\ee
where
\be
\hat{A}(\xi_i, \xi_j) = \begin{cases}
i_{\xi_i, \xi_j} \frac{1}{\xi_i-\xi_j} + A(\xi_i, \xi_j),  & i < j, \\
A(\xi_i, \xi_i),  & i =j, \\
i_{\xi_j, \xi_i} \frac{1}{\xi_i-\xi_j} + A(\xi_i, \xi_j),  & i > j,
\end{cases}
\ee
and
\be
\begin{split}
A(\xi, \eta)
= & \sum_{m,n \geq 0} A_{m,n}   \xi^{-n-1}
\eta^{-m-1} .
\end{split}
\ee
\end{thm}

This results enables to easily compute the $n$-point functions once
we have found the affine coordinates $A_{m,n}$.
For example,
\be
G^{(1)}(\xi) = \sum_{m, n\geq 0} A_{m,n} \xi^{m+n}.
\ee
Because there is only one $2$-cylce,
the two-point function is given by:
\ben
G^{(2)}(\xi_1, \xi_2)
& = &  \frac{A(\xi_1, \xi_2) - A(\xi_2,\xi_1)}{\xi_1-\xi_2}
-A(\xi_1, \xi_2)\cdot A(\xi_2,\xi_1).
\een
There are two $3$-cycles $(1,2,3)$ and $(1,3,2)$,
and so the three-point function is given by:
\ben
G^{(3)}(\xi_1, \xi_2,\xi_3)
& = &  A(\xi_1,\xi_2)A(\xi_2,\xi_3)A(\xi_3,\xi_1)
+ A(\xi_1,\xi_3)A(\xi_3,\xi_2)A(\xi_2,\xi_1) \\
& + & \frac{A(\xi_2, \xi_3)A(\xi_3,\xi_1) - A(\xi_1,\xi_3)A(\xi_3,\xi_2)}{\xi_1-\xi_2} \\
& + & \frac{A(\xi_3, \xi_1)A(\xi_1,\xi_2) - A(\xi_2,\xi_1)A(\xi_1,\xi_3)}{\xi_2-\xi_3} \\
& + & \frac{A(\xi_1, \xi_2)A(\xi_2,\xi_3) - A(\xi_3,\xi_2)A(\xi_2,\xi_1)}{\xi_3-\xi_1} \\
& + & \frac{A(\xi_3,\xi_1) + A(\xi_1, \xi_3)}{(\xi_1-\xi_2)(\xi_2-\xi_3)}
+ \frac{A(\xi_1,\xi_2) + A(\xi_2, \xi_1)}{(\xi_2-\xi_3)(\xi_3-\xi_1)} \\
& + & \frac{A(\xi_2,\xi_3) + A(\xi_3, \xi_2)}{(\xi_3-\xi_1)(\xi_1-\xi_2)}.
\een

\subsection{Principal specialization and quantum spectral curve}

Once the affine coordinates of $V$ in Sato Grassmannian $GR_0$ are known,
one gets the expansion of the tau-function $Z_V$ in expansion in Schur functions:
\be
Z_V:
= \sum_\mu \det(A_{m_i,n_j})_{1\leq i , j \leq k } \cdot
s_\mu.
\ee
From the point of view of the theory of symmetric functions,
this is of the form
\be \label{eqn:Sum}
\sum_{\mu \in \cP} s_\mu(\bx) s_\mu(\by) = \exp \sum_{n\geq 1} \frac{1}{n} p_n(\bx) p_n(\by),
\ee
where $\bx=(x_1, x_2, \dots)$, $\by=(y_1, y_2, \dots)$,
the coefficients $\det(A_{m_i,n_j})_{1\leq i , j \leq k }$ are
just some specialization of $s_\mu(\by)$,
specified by
\be \label{eqn:Spec1}
\epsilon_\by s_{(m|n)}(\by) = A_{m,n}.
\ee
One can also consider the specialization of $s_\mu = s_\mu(\bx)$.
The {\em principal specialization} is specified by:
\be \label{eqn:Spec2}
\epsilon_\bx p_n(\bx) = z^{n}.
\ee
One can apply the specialization \eqref{eqn:Spec2} on both sides of \eqref{eqn:Sum} to get:
\ben
&& \sum_{\mu \in \cP}   \epsilon_\bx(s_\mu(\bx)) \cdot s_\mu(\by)
= \epsilon_\bx
\exp \sum_{n\geq 1} \frac{1}{n} p_n(\bx) p_n(\by) \\
& = & \exp \sum_{n\geq 1} \frac{z^n}{n}   p_n(\by)
= 1 + \sum_{k \geq 1} h_k(\by) z^k \\
& = & 1 + \sum_{m \geq 0} z^{m+1} s_{(m|0)}(\by),
\een
where $h_k(\by)$ are the complete symmetric functions in $\by$.
Therefore,
after also applying $\epsilon_\by$,
one gets:
\be
\epsilon_\bx Z_V = 1 + \sum_{m\geq 0} A_{m,0} z^{m+1} = \Phi_1(z^{-1}).
\ee
This has been observed by Alexandrov \cite{Alex-BGW}.

Once we know $\Phi_1(z)$,
one can look for differential equation of the form
\be
\Delta \Phi_0(z) = 0.
\ee
This is called the {\em quantum spectral curve}.
One can also look for  a differential   operator $\Delta$ such that
\be
\Delta \Span \{\Phi_n(z)\}_{n \geq 0} = \Span \{\Phi_n(z)\}_{n \geq 0}.
\ee
This is called a {\em Kac-Schwarz operator}.
The most interesting Kac-Schwarz operator is such that the lead term of $\Delta^n\Phi_1(z)$
has degree $n$,
so that they can be used to form a basis.

\section{Counting Grothendieck's Dessins D'Endfants as KP Tau-Function}

In this Section we prove the explicit formula for the counting function
of Grothendieck's dessins d'enfants as a KP tau-function.
We also reformulate the quantum spectral curve in this case as
differential equation of hypergeometric type,
and derive some Kac-Schwarz operators.

\subsection{The dessin tau-function}

Denote by $F_{Dessins}(s, u, v, p_1, p_2, \dots )$ the generating series of weighted count
of labelled dessins d'enfants defined in \cite{Zog} and \cite{Kaz-Zog}:
\be
\begin{split}
& F_{Dessins}(s, u, v, p_1, p_2, \dots ) \\
= &\sum_{k,l,m\geq 1}\frac{1}{m!} \sum_{\mu_1, \dots, \mu_m \geq 1}
N_{k,l}(\mu_1, \dots, \mu_m)s^du^kv^l p_{\mu_1} \cdots  p_{\mu_m},
\end{split}
\ee
where $d = \sum_{i=1}^m \mu_i$.
Let $Z_{Dessins} = e^{F_{Dessins}}$.
Then by \cite{Zog} and \cite{Kaz-Zog},
$Z_{Dessins}$ is a three-parameter family of tau-functions of the KP hierarchy for fixed $s$, $u$ and $v$,
and furthermore,
it satisfies the following sequence of Virasoro constraints:
\be
L_n Z_{Dessins} = 0,
\ee
where for $n \geq 0$,
\be
\begin{split}
L_n  = & -\frac{n + 1}{s} \frac{\pd}{\pd p_{n+1}}
+(u +v)n \frac{\pd}{\pd p_n}
+ \sum_{j=1}^\infty p_j(n + j) \frac{\pd}{\pd p_{n+ j}} \\
+ & \sum_{i + j=n} i j \frac{\pd^2}{\pd p_i\pd p_j}
+\delta_{n,0}uv.
\end{split}
\ee
The following Virasoro commutation relations are satisfied by the operators $L_n$ for $m, n \geq 0$:
\be
[L_m, L_n] = (m - n)L_{m+n}.
\ee
Furthermore,
as a corollary,
\be \label{eqn:Z-CJ}
Z_{Dessins}(u,v,s) = e^{s((u+v)\Lambda_1 + M_1 + uvp_1)}1,
\ee
where $\Lambda_1$ and $M_1$ are differential operators defined as follows:
\bea
&& \Lambda_1 = \sum_{i=2}^\infty (i - 1)p_i \frac{\pd}{\pd p_{i-1}}, \\
&& M_1 = \sum_{i=2} \sum^{i-1}_{j=1}
\biggl((i - 1)p_jp_{i-j} \frac{\pd}{\pd p_{i-1}} + j(i - j)p_{i+1} \frac{\pd^2}{\pd p_j \pd p_{i-j}} \biggr).
\eea
We will refer to $Z_{Dessins}$ as the dessin tau-function.

\subsection{Reformulation of $Z_{Dessins}$ in terms of Schur functions}
By \eqref{eqn:Z-CJ} we have
\ben
Z_{Dessins} & = & 1 + suv p_1 \\
& + & \frac{1}{2}s^2\biggl((u+v) uv p_2 + (uv+u^2v^2)p_1^2\biggr) \\
& + & \frac{s^3uv}{6} \biggl((2u^2+6uv+2v^2+2) p_3 + 3(u+v)(uv+2)p_2p_1 \\
&& + (uv+1)(uv+2) p_1^3 \biggr) \\
& + & \frac{s^4uv}{24}  \biggl( 6(u+v)(u^2+5uv+v^2+5) p_4 \\
&& + 8(uv+3)(u^2+3uv+v^2+1) p_3p_1 \\
&& + 3(u^3v+2u^2v^2+uv^3+4u^2+10uv+4v^2+2) p_2^2 \\
&& + 6 (u+v)(uv+2) (uv+3) p_2p_1^2 \\
&& + (uv+1)(uv+2)(uv+3)p_1^4 \biggr) + \cdots.
\een

We use the formula \eqref{eqn:Newton-Schur} to rewrite $Z_{Dessins}$ in terms of Schur functions:
\ben
Z & = & 1 + suv s_{(1)} \\
& + & \frac{s^2u(u+1)v(v+1)}{2} s_{(2)}
+ \frac{s^2u(u-1)v(v-1)}{2} s_{(1^2)} \\
& + &  \frac{s^3u(u+1)(u+2)v(v+1)(v+2)}{6} s_{(3)} \\
& + & \frac{s^3u(u+1)(u-1)v(v+1)(v-1)}{3} s_{(2,1)} \\
& + & \frac{s^3u(u-1)(u-2)v(v-1)(v-2)}{6} s_{(1^3)}   \\
& + & \frac{s^4u(u+1)(u+2)(u+3)v(v+1)(v+2)(v+3)}{24} s_{(4)} \\
& + & \frac{s^4u(u+1)(u+2)(u-1)v(v+1)(v+2)(v-1)}{8} s_{(3,1)} \\
&+ &  \frac{s^4u^2(u+1)(u-1)v^2(v+1)(v-1)}{12} s_{(2^2)} \\
& + & \frac{s^4u(u+1)(u-1)(u-2)v(v+1)(v-1)(v-2)}{8} s_{(2,1^2)} \\
& + & \frac{s^4u(v+1)v(v-1)(v-2)(v-3)}{24} s_{(1^4)}  + \cdots.
\een

Now we state our main theorem.
\begin{thm} \label{thm:Main1}
The partition function defined by weighted count of Grothendieck's dessins d'enfants
is given explicitly by the following formula:
\be
Z_{Dessin} = \sum_{\mu \in \cP} s_\mu \cdot \prod_{e \in \mu} \frac{s \cdot (u+c(e))(v+c(e))}{h(e)},
\ee
where $e=(i,j)$ denotes the box at the $i$-th row and the $j$-column when we represent the partition $\mu$
by its Young diagram,
$c(e)$ and $h(e)$ denotes its content and hook length respectively \cite{Macdonald}.
\end{thm}

This follows from the following:

\begin{thm} \label{thm:Main2}
In the fermionic picture,
the dessin tau-function is given by a Bogoliubov transformation:
\be \label{eqn:Z-A}
Z_{Dessins} = \exp(\sum_{m,n \geq 0} A_{m,n} \psi_{-m-\frac{1}{2}}\psi^*_{-n-\frac{1}{2}}) \vac,
\ee
where   the coefficients $A_{m,n}$ as explicitly given
as follows:
\be \label{eqn:Amn}
A_{m, n}
= \frac{(-1)^n s^{m+n+1}uv}{(m+n+1)m!n!}
\prod_{j=1}^{m} (u+j)(v+j) \cdot \prod_{i=1}^n (u-i)(v-i).
\ee
\end{thm}

The proof will be presented in the Appendix.

\subsection{The free energy and the $n$-point functions}

The first few terms of the dessin free energy $F_{Dessins}$  are:
\ben
F_{Dessins} & = &  suv p_1
+ s^2uv \biggl((u+v) \frac{p_2}{2} + \frac{p_1^2}{2}  \biggr)  \\
& + & s^3uv \biggl((u^2+3uv+v^2+1) \frac{p_3}{3} + (u+v)p_2p_1 +   \frac{p_1^3}{3} \biggr) \\
& + & s^4uv  \biggl( (u+v)(u^2+5uv+v^2+5) \frac{p_4}{4} \\
&+ &(u^2+3uv+v^2+1) p_3p_1 \\
&+ & \frac{1}{4}(2u^2+5uv+2v^2+1) p_2^2
+ \frac{3}{2}(u+v) p_2p_1^2 + \frac{p_1^4}{4} \biggr) + \cdots.
\een
The $n$-point functions are defined by:
\be
G^{(n)}(\xi_1, \dots, \xi_n)
= \sum_{m_1, \dots, m_n\geq 1}
\prod_{j=1}^n \frac{m_1}{\xi_j^{-m_j-1}} \frac{\pd}{\pd p_{m_j}}  F_{Dessins} \biggl|_{p_m =0, m \geq 1}
\cdot
\ee
For example,
\ben
&& G^{(1)}(\xi)
= suv \xi^{-2} + s^2uv(u+v) \xi^{-3}
+ s^3uv(u^2+3uv+v^2+1) \xi^{-4} + \cdots, \\
&& G^{(2)}(\xi_1, \xi_2) = s^2uv \xi_1^{-2}\xi_2^{-2} + 2s^3uv(u+v) (\xi_1^{-3}\xi_2^{-2}
+ \xi_1^{-2} \xi_2^{-3}) + \cdots.
\een
One can use our formula for $n$-point functions associated with
tau-functions of KP hierarchy \cite{Zhou-Emerg} to compute the $n$-point functions.
The following are the first few terms of $A(\xi, \eta)$:
\ben
A(\xi, \eta)
& = & suv \cdot \xi^{-1}\eta^{-1} \\
& + & \half s^2u(u+1)v(v+1) \cdot \xi^{-1}\eta^{-2}
- \half s^2u(u-1)v(v-1) \cdot \xi^{-2} \eta^{-1} \\
& + & \frac{1}{6}s^3 u(u+1)(u+2)v(v+1)(v+2) \cdot \xi^{-1} \eta^{-3} \\
& - & \frac{1}{3}s^3 u(u+1)(u-1)v(v+1)(v-1) \cdot \xi^{-2} \eta^{-2} \\
& + & \frac{1}{6}s^3 u(u-1)(u-2)v(v-1)(v-2) \cdot \xi^{-3} \eta^{-1}  + \cdots.
\een
The one-point function is then
\ben
G^{(1)}(\xi) & = &
\sum_{n \geq 1} \frac{s^nuv}{n} \sum_{i+j=n-1}
 \frac{(-1)^j}{i!j!} \prod_{a=1}^{i} (u+a)(v+a)  \\
 && \cdot \prod_{b=1}^j (u-b)(v-b) \cdot \xi^{-n-1}.
\een

\subsection{Quantum spectral curve for dessin counting as hypergeometric equation}
The principal specialization of the dessin partition function has been studied in \cite{Kaz-Zog}:
\be
\psi = \psi(s, t, u, v) = e^{F_{Dessins}(s,u,v,p_1,p_2,...)}|_{ p_i=t^i} ,
\ee
where $t$ is a new formal variable.
By our result on the affine coordinates of the dessin tau-function above,
we have:
\be
\psi = 1 + \sum_{n=1}^\infty \frac{s^n t^n}{n!} \prod_{j=0}^{n-1} (u+j)(v+j).
\ee
If we write  $\psi = \sum_{n=0}^\infty a_n t^n$,
then we have
\bea
&& a_0  = 1, \\
&& a_n = \frac{s(u+n)(v+n)}{n} a_{n-1},
\eea
and so it is clear that $\psi$ satisfies the following equation of hypergeometric type:
\be
\frac{\pd}{\pd t} \psi = s (t\frac{\pd}{\pd t} + u)(t\frac{\pd}{\pd t} + v) \psi.
\ee
It can be  rewritten as the Schr\"odinger equation:
\be
t^2 \frac{d^2\psi}{dt^2} + \biggl((u + v + 1) t - \frac{1}{s} \biggr) \frac{d\psi}{ dt} + uv\psi = 0.
\ee
This recovers the quantum spectral curve  for dessin counting
in \cite{Kaz-Zog}.

\subsection{A Kac-Schwarz operator for dessin counting}

Now we derive some Kac-Schwarz operator for the dessin tau-function.
By our formula for the affine coordinates $A_{m,n}$ of the dessin tau-function,
if we define
\be
\phi_{k,Des }(z;u,v,s) =z^k + \sum_{l\geq 0} A_{l,k}(u,v,s) z^{-l-1},
\ee
then
\ben
&& \phi_{k,Des }(z;u,v,s) \\
& = & z^k + \sum_{n=k+1}^\infty \frac{(-1)^ks^n}{n\cdot  k!(n-k-1)!z^{n-k}}
\prod_{j=0}^{n-1} (u-k+j)(v-k+j) \\
& = & \frac{z^k}{k!} \frac{d^k}{dz^k}
\biggl(z^k \sum_{n=0}^\infty \frac{s^n}{n!z^n}
\prod_{j=0}^{n-1} (u-k+j)(v-k+j) \biggr).
\een
In particular,
\be
\phi_{0,Des }(z; u,v,s)= 1 + \sum_{n=1}^\infty \frac{s^n}{n!z^n} \prod_{j=0}^{n-1} (u+j)(v+j).
\ee
One can use $\Gamma$-function to give it an integral representation:
\ben
\phi_{0,Des }(z;u,v,s) & = & \sum_{n=0}^\infty \frac{(-1)^ns^n}{z^n} \binom{-u}{n} \frac{\Gamma(v+n)}{\Gamma(v)} \\
& = & \frac{1}{\Gamma(v)}  \int_0^\infty e^{-x} \sum_{n=0}^\infty
\frac{(-1)^ns^n}{z^n}\binom{-u}{n} x^{v+n-1} dx   \\
& = & \frac{1}{\Gamma(v)} \int_0^\infty e^{-x} (1-\frac{sx}{z})^{-u}x^{v-1} dx \\
& = & \frac{z^v}{(-s)^v\Gamma(v)} \int_0^\infty e^{zx/s} (1+x)^{-u} x^{v-1} dx.
\een
For later use£¬
note:
\ben
&& \frac{\pd}{\pd z}\phi_{0,Des }(z; u,v,s) \\
& = &  \frac{v}{z} \cdot \frac{z^v}{(-s)^v\Gamma(v)} \int_0^\infty e^{zx/s} (1+x)^{-u} x^{v-1} dx \\
& + & \frac{1}{s} \frac{z^v}{(-s)^v\Gamma(v)} \int_0^\infty x\cdot e^{zx/s} (1+x)^{-u} x^{v-1} dx \\
& = & \frac{v}{z} \cdot \phi_{0,Des}(z;u,v,s)
+ \frac{1}{s} \frac{z^v}{(-s)^v\Gamma(v)} \int_0^\infty x\cdot e^{zx/s} (1+x)^{-u} x^{v-1} dx.
\een
To find a Kac-Schwarz operator,
first note:
\ben
&& \phi_{1,Des}(z; u, v,s) \\
& = & z - \sum_{n=2}^\infty \frac{s^n}{n\cdot  (n-2)!z^{n-1}} \prod_{j=0}^{n-1} (u-1+j)(v-1+j) \\
& = & z \frac{d}{dz} \biggl(z \sum_{n=0}^\infty \frac{s^n}{n!z^n} \prod_{j=0}^{n-1} (u-1+j)(v-1+j) \biggr).
\een
Because we have
\ben
&& \sum_{n=0}^\infty \frac{s^n}{n!z^n} \prod_{j=0}^{n-1} (u-1+j)(v-1+j) \\
& = &  \frac{z^{v-1}}{(-s)^{v-1}\Gamma(v-1)} \int_0^\infty e^{zx/s} (1+x)^{-u+1} x^{v-2} dx \\
& = &   \frac{z^{v-1}}{(-s)^{v-1}(v-1)\Gamma(v-1)} \int_0^\infty e^{zx/s} (1+x)^{-u+1} dx^{v-1}  \\
& = &   - \frac{z^{v-1}}{(-s)^{v-1}\Gamma(v)} \int_0^\infty x^{v-1} d(e^{zx/s} (1+x)^{-u+1})     \\
& = &  - \frac{z^{v-1}}{(-s)^{v-1}\Gamma(v)} \int_0^\infty
\biggl(\frac{z}{s} e^{zx/s} (1+x)^{-u+1} -(u-1) e^{zx/s}(1+x)^{-u}\biggr)  x^{v-1}dx    \\
& = &  \phi_{0,Des}(z;u,v,s) - (u+v-1)\frac{s}{z}\phi_{0,Des}(z;u,v,s) + s \frac{\pd}{\pd z} \phi_{0,Des}(z;u,v,s).
\een
It follows that
\ben
&& \phi_{1,Des}(z;u,v,s) \\
& = & z\frac{d}{dz} \biggl[z(
\phi_{0,Des}(z;u,v,s) - (u+v-1)\frac{s}{z}\phi_{0,Des}(z;u,v,s) \\
&& + s \frac{\pd}{\pd z} \phi_{0,Des}(z;u,v,s))\biggr] \\
& = & \cD \phi_{0,Des}(z;u,v,s),
\een
where
\ben
\cD & = & z^2\frac{d}{d z} + z - (u+v-2) s z \frac{d}{dz}
+ sz^2 \frac{d^2}{dz^2}.
\een
Note:
\ben
&& \cD z^m = (m+1)z^{m+1} -ms(u+v-m-1) z^m, \\
&& \cD \frac{1}{z^n}
= - \frac{n-1}{z^{n-1}}+ \frac{ns(u+v+n-1)}{z^n}.
\een

For simplicity of notations,
we write $\phi_{k,Des}(z;u,v,s)$ as $\phi_{k, Des}(z;u,v,s)$.

\begin{thm}
The operator $\cD$ is a Kac-Schwarz operator for the dessin tau-function:
\be
\cD\phi_{k,Des}(z)
= (k+1)\phi_{k+1,Des}(z) -ks(u+v-k-1) \phi_{k,Des}(z).
\ee
\end{thm}

\begin{proof}
We rewrite $\phi_{k,Des}(z)$ as follows:
\ben
&& \phi_{k,Des}(z)
= z^k + \sum_{n=1}^\infty \frac{(-1)^ks^{n+k}}{(n+k)\cdot  k!(n-1)!z^{n}}
\prod_{j=0}^{n+k-1} (u-k+j)(v-k+j).
\een
\ben
&& \cD\phi_{k,Des}(z) \\
& = &  (k+1)z^{k+1} -ks(u+v-k-1) z^k \\
& + & \sum_{n=1}^\infty \frac{(-1)^ks^{n+k}}{(n+k)\cdot  k!(n-1)!}
\cdot \biggl( - \frac{n-1}{z^{n-1}}+ \frac{ns(u+v+n-1)}{z^{n}}
\biggr) \\
&& \cdot \prod_{j=0}^{n+k-1} (u-k+j)(v-k+j) \\
& = &  (k+1)\phi_{k+1,Des}(z) -ks(u+v-k-1) \phi_{k,Des}(z) \\
& + & (k+1) \cdot \sum_{n=1}^\infty \frac{(-1)^ks^{n+k+1}}{(n+k+1)\cdot  (k+1)!(n-1)!z^{n}}
\prod_{j=0}^{n+k} (u-k-1+j)(v-k-1+j) \\
& + & ks(u+v-k-1) \cdot  \sum_{n=1}^\infty \frac{(-1)^ks^{n+k}}{(n+k)\cdot  k!(n-1)!z^{n}}
\prod_{j=0}^{n+k-1} (u-k+j)(v-k+j) \\
& + & \sum_{n=1}^\infty \frac{(-1)^ks^{n+k}}{(n+k)\cdot  k!(n-1)!}
\cdot \biggl( - \frac{n-1}{z^{n-1}}+ \frac{ns(u+v+n-1)}{z^{n}}
\biggr) \\
&& \cdot \prod_{j=0}^{n+k-1} (u-k+j)(v-k+j).
\een
\ben
&& \cD\phi_{k,Des}(z) -  (k+1)\phi_{k+1,Des}(z)  + ks(u+v-k-1) \phi_{k,Des}(z) \\
& = & \sum_{n=1}^\infty \frac{(-1)^ks^{n+k+1}}{(n+k+1)\cdot k!(n-1)!z^{n}}
\prod_{j=0}^{n+k-1} (u-k+j)(v-k+j) \\
&& \cdot \biggl( (u-k-1)(v-k-1)-(u+n)(v+n) \biggr) \\
& + &  \sum_{n=1}^\infty \frac{(-1)^ks^{n+k+1} }{(n+k)\cdot  k!(n-1)!z^{n}}
\prod_{j=0}^{n+k-1} (u-k+j)(v-k+j) \\
&& \cdot \biggl(k(u+v-k-1)+n(u+v+n-1) \biggr) \\
& = & - \sum_{n=1}^\infty \frac{(-1)^ks^{n+k+1}}{(n+k+1)\cdot k!(n-1)!z^{n}}
\prod_{j=0}^{n+k-1} (u-k+j)(v-k+j) \\
&& \cdot  (n+k+1)(u+v+n-k-1) \\
& + &  \sum_{n=1}^\infty \frac{(-1)^ks^{n+k+1} }{(n+k)\cdot  k!(n-1)!z^{n}}
\prod_{j=0}^{n+k-1} (u-k+j)(v-k+j) \\
&& \cdot (n+k)(u+v+n-k-1) \\
& = & 0.
\een
\end{proof}

\section{Counting Clean Dessins and Hermitian One-Matrix Model}

In this Section we repeat the same discussions for counting clean dessins,
which we recall is equivalent to Hermitian one-matrix model,
and so our earlier results on the latter can be applied.
We present a simple relationship between the affine coordinates
of the dessin tau-function and the clean dessin tau-function.
We also discuss the quantum spectral curve and Kac-Schwarz operators in this case.

\subsection{Counting clean dessins}

Let us now recall some results due to Kazarian-Zograf on counting clean dessins \cite{Kaz-Zog}.
We will follow their notations.
Denote by $D_{g,m}(\mu) = D_{g,m}(\mu_1, . . . , \mu_m)$
the number of genus $g$ ribbon graphs
with $m$ labeled vertices of degrees
$\mu_1,\dots, \mu_m$ counted with weights
$\frac{1}{|Aut_v \Gamma|}$,
where the automorphisms preserve each vertex of $\Gamma$ pointwise.
The free energy in this case is the generating function
\be
\tilde{F}_{Dessins}(s, u, p_1, p_2, \dots) =
\sum_{g=0}^\infty \sum_{m=1}^\infty
\frac{1}{m!} \sum_{\mu\in \bZ^m_+}
D_{g,m}(\mu)s^du^k \prod_{j=1}^m  p_{\mu_j},
\ee
where $d = \sum^m_{i=1} \mu_i$, $k = 2 -2g - m+ d/2$,
and $\mu = (\mu_1, \dots, \mu_m)$.
The partition function is then $\tilde{Z}_{Dessins} = e^{\tilde{F}_{Dessins}}$.
It is a tau function of the KP hierarchy
 for any $s$ and $u$,
uniquely determined by the Virasoro constraints:
\be
\tilde{L}_n \tilde{Z}_{Dessins} = 0,
\ee
where the operator $\tilde{L}_n$ is defined by:
\be
\begin{split}
\tilde{L}_n = & -\frac{n + 2}{s^2}
\frac{\pd}{\pd p_{n+2}} + 2 u n \frac{\pd}{\pd p_n}
+ \sum_{j=1}^\infty p_j(n + j) \frac{\pd}{\pd p_{n+j}} \\
+&  \sum_{i+j=n} ij \frac{\pd^2}{\pd p_i\pd p_j}
+ \delta_{n,-1}up_1 + \delta_{n,0}u^2,
\end{split}
\ee
where $n \geq -1$.
Furthermore,
$\tilde{Z}_{Dessins}$  is explicitly given by the following cut-and-join representation:
\be
 \tilde{Z}_{Dessins}
= e^{ \frac{s^2}{2} (2 u \Lambda_2+M_2+u^2p_2)}1,
\ee
where $\Lambda_2$ and $M_2$ are the following operators:
\be
\Lambda_2 =
\sum_{i=1}^\infty
i p_{i+2} \frac{\pd}{\pd p_i}
+ \frac{1}{2}p^2_1,
\ee
\be
M_2 =
\sum_{i=2}^{\infty} \sum^{i-1}_{j=1}
\biggl( (i - 2)p_jp_{i-j} \frac{\pd}{\pd  p_{i-2}}
+ j(i - j)p_{i+2} \frac{\pd^2}{\pd p_j\pd p_{i-j}} \biggr).
\ee
We will refer to this tau-function as the {\em clean dessin tau-function}.

\subsection{Equivalence with Hermitian one-matrix model}

Let us now recall that the counting of ribbon graphs is equivalent to the Hermitian one-matrix model.
Recall that if  one introduces the 't Hooft coupling constant $t = Ng_s$,
then the partition functions of Hermitian matrix models satisfy the Virasoro constraints
with  Virasoro operators given by:
\ben
L_{-1,t} &=& - \frac{\pd}{\pd g_1}
+ \sum_{n \geq 1} ng_{n+1} \frac{\pd}{\pd g_n} +  tg_1g_s^{-2}, \\
L_{0,t} &=& - 2\frac{\pd}{\pd g_2} + \sum_{n \geq 1} ng_n \frac{\pd}{\pd g_n}
+t^2g_s^{-2}, \\
L_{1,t} &= &- 3\frac{\pd}{\pd g_3} + \sum_{n \geq 1} (n+1)g_n \frac{\pd}{\pd g_{n+1}}
+ 2t\frac{\pd}{\pd g_1}, \\
L_{m,t} & = &  \sum_{k \geq 1} (k+m) (g_k-\delta_{k,2}) \frac{\pd}{\pd g_{k+m}} \\
&& + g_s^2 \sum_{k=1}^{m-1} k(m-k)\frac{\pd}{\pd g_k} \frac{\pd}{\pd g_{m-k}}
+ 2 tm  \frac{\pd}{\pd g_m},
\een
where $m \geq 2$.
These are called  the {\em fat Virasoro constraints} in \cite{Zhou-MM2}.
It is clear that if we set
\begin{align}
g_s & = s^2, & g_k & = g_s p_k, & t = g_s u,
\end{align}
then we recover the Virasoro constraints for clean dessins.
So we can identify the clean dessin tau-funtion
with the partition function of Hermitian one-matrix model
after the above identification of coupling constants.

\subsection{Fermionic representation of clean dessin tau-function}

Recall one of the main results in \cite{Zhou-MM} is the following formula
for the partition functions of Hermitian one-matrix models of size $N \times N$:
\be
Z_N|_{g_s=1}
= \sum_\lambda c(\lambda) \cdot \prod_{x\in \lambda}  \frac{(N+c(x))}{h(x)} \cdot s_\lambda.
\ee
where $c(\lambda)$ is defined by:
\be
c(\lambda):=(2n-1)!! \frac{\chi^\lambda_{(2^n)}}{\chi^\lambda_{(1^{2n})}}.
\ee
From this one derives the following

\begin{thm}
The partition function of counting clean dessins is explicitly given by the following formula:
\be \label{eqn:Z-clean}
\tilde{Z}_{Dessins} = \sum_\lambda c(\lambda) \cdot s^{|\lambda|}
\prod_{x\in \lambda}  \frac{(u+c(x))}{h(x)} \cdot s_\lambda,
\ee
where the summation is taken over all partitions $\lambda$ of even weight $|\lambda|$.
\be
\tilde{Z}_{Dessins} = \exp \biggl( \sum_{n=1}^\infty \sum_{p=0}^{2n-1}
\tilde{A}_{2n-1-p,p}(u,s) \cdot \psi_{-(2n-p)+1/2} \psi_{-p-1/2}^*
 \biggr) \vac,
\ee
where the coefficients $\tilde{A}_{m,n}$ are explicitly given by:
\ben
\tilde{A}_{2n-1-p,p}(u,s)
& = & \frac{(2n-1)!!}{(2n)!}
\cdot (-1)^{p+[(p+1)/2]} \binom{n-1}{[p/2]}
 \cdot s^{2n} \cdot [u]_{-p}^{2n-1-p}.
\een
\end{thm}

Here we have used the following notation:
For two integers $m < n$,
\be
[x]_m^n = \prod_{j=m}^n (x+j).
\ee
By \eqref{eqn:Amn},
the corresponding coefficients for the dessin tau-function can be rewritten in this notation as follows:
\be
A_{j, k}(u,v,s)
= \frac{(-1)^k s^{j+k+1}}{(j+k+1)j!k!}[u]_{-k}^j \cdot [v]_{-k}^j.
\ee

\begin{thm} \label{thm:A-A}
The coefficients $\tilde{A}_{2n-k,k}$ are related to $A_{j,k}$ in the following way:
\be
\tilde{A}_{2n-2k-1,2k}(u,s) = A_{n-1-k,k}(u/2, (u+1)/2, 2s^2),
\ee
and
\be
\tilde{A}_{2n-2k-2,2k+1}(u,s) = A_{n-1-k,k}(u/2, (u-1)/2, 2s^2).
\ee
\end{thm}

\begin{proof}
The first identity can be checked as follows:
\ben
&& \tilde{A}_{2n-2k-1,2k}(u,s) \\
& = & \frac{(2n-1)!!}{(2n)!}
\cdot (-1)^{2k+[(2k+1)/2]} \binom{n-1}{[2k/2]}
 \cdot s^{2n} \cdot [u]_{-2k}^{2n-1-2k} \\
& = & \frac{(-1)^ks^{2n}}{2^nn!}\cdot \frac{(n-1)!}{(n-1-k)!k!}
\prod_{l=-2k}^{2n-1-2k} (u+l) \\
& = & \frac{(-1)^k2^ns^{2n}}{n\cdot (n-1-k)!k!}\cdot
\prod_{l=-k}^{n-1-k} (\frac{u}{2}+l) (\frac{u+1}{2}+l) \\
& = & A_{n-1-k,k}(u/2, (u+1)/2, 2s^2).
\een
Similarly,
\ben
&& \tilde{A}_{2n-2k-2,2k+1}(u,s) \\
& = & \frac{(2n-1)!!}{(2n)!}
 \cdot (-1)^{2k+1+[(2k+2)/2]} \binom{n-1}{[(2k+1)/2]}
 \cdot s^{2n} \cdot [u]_{-2k-1}^{2n-2k-2} \\
& = & \frac{(-1)^ks^{2n}}{2^nn!}\cdot \frac{(n-1)!}{(n-1-k)!k!}
\prod_{l=-2k-1}^{2n-2-2k} (u+l) \\
& = & \frac{(-1)^k2^ns^{2n}}{n\cdot (n-1-k)!k!}\cdot
\prod_{l=-k}^{n-1-k} (\frac{u}{2}+l) (\frac{u-1}{2}+l) \\
& = & A_{n-1-k,k}(u/2, (u-1)/2, 2s^2).
\een
\end{proof}

As a corollary we have

\begin{thm}
If we define
\be
\tilde{\phi}_{k,Des }(z;u) =z^k + \sum_{l\geq 0} \tilde{A}_{l,k}(u,v) z^{-l-1}, \;\;\; k \geq 0,
\ee
then they are related to $\phi_{k,Des}(z;u,v)$ as follows:
\bea
&& \tilde{\phi}_{2k,Des}(z,u) = \phi_k(z^2;u/2, (u+1)/2, 2s^2), \\
&& \tilde{\phi}_{2k-1,Des}(z,u) = z \cdot \phi_k(z^2;u/2, (u-1)/2, 2s^2).
\eea
\end{thm}

\subsection{Quantum spectral curve for clean dessins as hypergeometric equation}

Similar to the case of dessin tau-function,
the principal specialization of the clean dessin tau-function is
\be
\tilde{\psi} = \tilde{\psi}(s, q, u) = e^{\tilde{F}(s,u,p_1,p_2,...)}|_{p_i=q^i} ,
\ee
By \eqref{eqn:Z-clean} we easily get:
\ben
\tilde{\psi} 
& = &  1 + \sum_{n=1}^\infty \frac{(2n-1)!!}{(2n)!} s^{2n} q^{2n} \prod_{j=0}^{2n-1} (u+j) \\
& = & 1 + \sum_{n=1}^\infty \frac{s^{2n} x^{n}}{ n!}  \prod_{j=0}^{2n-1} (u+j),
\een
where $x= q^2/2$.
From this it is clear that the following equation of hypergeometric type
is satisfied by $\tilde{\psi}$:
\be
 \frac{d}{dx} \tilde{\psi} = s^2 (2x\frac{d}{dx} + u)(2x\frac{d}{dx} + u+1) \tilde{\psi}.
\ee
Since $x\frac{d}{dx} = \half q\frac{d}{dq}$,
this equation can be rewritten as follows:
\be
q^2 \frac{d^2\tilde{\psi}}{dq^2} + \biggl(2(u + 1) q - \frac{1}{s^2q} \biggr) \frac{d}{dq} \tilde{\psi}
 + (u + u^2) \tilde{\psi} = 0 .
\ee
This recovers the quantum spectral curve for counting ribbon graphs in \cite{Kaz-Zog}.

\section{Dessins with Even Vertices and Modified Partition Function of Hermitian One-Matrix Model with Even Couplings}

In this Section we consider
counting of clean dessins with only vertices of even valences.
It is related to modified partition of Hermitian one-matrix model with even couplings,
introduced in \cite{DLYZ}.
We show that the latter is a special case of the dessin tau-function
which count all possible dessins.

\subsection{Modified GUE partition function and modified Virasoro constraints}

Let $Z_{even}$ denote the GUE partition function with even couplings.
Then it is a generating series that enumerates clean dessins with vertices of only even vertices.
The modified partition function $\tilde{Z}_{even}$
is introduced in \cite{DLYZ}.
It is defined by
\be
\log Z_{even} = \big( \Lambda^{1/2} + \Lambda^{-1/2} \big)  \log \tilde{Z}_{even}.
\ee
It is proved in \cite{DLYZ} that
$\tilde{Z}$ satisfies the followings system of
Virasoro constraints:
\ben
&& \half \frac{\pd}{\pd s_2}\tilde{Z}_{even} =  \sum_{k\geq 1}
k s_{2k} \frac{\pd}{\pd s_{2k}} \tilde{Z}_{even}
+ \biggl(\frac{t^2}{4\epsilon^2}- \frac{1}{16}
\biggr) \tilde{Z}_{even},   \\
&& \half \frac{\pd}{\pd s_{2n+2}} \tilde{Z}_{even}=
  \sum^{n-1}_{k=1} \frac{\pd^2}{\pd s_{2k} \pd s_{2n-2k}} \tilde{Z}_{even}
+ t \frac{\pd}{\pd s_{2n}} \tilde{Z}_{even} +  \sum_{k\geq 1}
ks_{2k} \frac{\pd}{\pd s_{2k+2n}}\tilde{Z}_{even},
\een
$n \geq 1$,
where $t=N\epsilon$, and $\Lambda = e^{\epsilon \pd_t}$.

\subsection{Cut-and-join representation}

If we make the following change of variables:
\be
s_{2k} = \frac{p_k}{k}, \qquad k \geq 1,
\ee
then the Virasoro constraints for $\tilde{Z}_{even}$ take the following form:
\ben
&& \half \frac{\pd}{\pd p_1}\tilde{Z}_{even} =  \sum_{k\geq 1}
k p_{k} \frac{\pd}{\pd p_{k}} \tilde{Z}_{even}
+ \biggl(\frac{t^2}{4\epsilon^2}- \frac{\delta_{g,1}}{16}
\biggr) \tilde{Z}_{even},    \\
&& \half (n+1) \frac{\pd}{\pd p_{n+1}} \tilde{Z}_{even}=
\epsilon^2 \sum^{n-1}_{k=1} k(n-k) \frac{\pd^2}{\pd p_{k} \pd p_{n-k}} \tilde{Z}_{even}
+ nt \frac{\pd}{\pd p_{n}} \tilde{Z}_{even} \\
&& \qquad \qquad +  \sum_{k\geq 1}
(k+n)p_{k} \frac{\pd}{\pd p_{k+n}}\tilde{Z}_{even}.
\een
From this we derive the following equation:
\be \label{eqn:WZ}
\sum_{m=1}^\infty
mp_{m} \frac{\pd}{\pd p_{m}} \tilde{Z}_{even}
= W\tilde{Z}_{even},
\ee
where $W$ is an operator defined as follows:
\ben
W & = & 2p_1 \biggl(\frac{t^2\epsilon^{-2}}{4}- \frac{1}{16} \biggr)
+ 2\epsilon^2 \sum_{n=2}^\infty p_{n+1} \sum^{n-1}_{k=1} k (n-k)
\frac{\pd^2}{\pd p_{k} \pd p_{n-k}}  \\
& + & 2t \sum_{n=1}^\infty n p_{n+1}\frac{\pd}{\pd p_{n}}
+ 2\sum_{n=0}^\infty p_{n+1} \sum_{k\geq 1}
(k+n)p_{k} \frac{\pd}{\pd p_{k+n}},
\een
Now if we introduce a grading
\be
\deg p_{m} = m
\ee
and decompose $\tilde{Z}$ into homogeneous parts according to this grading:
\be
\tilde{Z}_{even} = 1 + \tilde{Z}^{(1)} + \tilde{Z}^{(2)} + \cdots,
\ee
Then \eqref{eqn:WZ} can be rewritten as a sequence of equations:
\be
\tilde{Z}^{(n)} = \frac{1}{n} W \tilde{Z}^{(n-1)},
\qquad n \geq 1,
\ee
with $\tilde{Z}^{(0)} = 1$.
Therefore,
one gets the following:

\begin{thm}
The modified partition function has the following cut-and-join representation:
\be
\tilde{Z}_{even} = e^W 1.
\ee
\end{thm}
The following are the first few terms of $\tilde{Z}^{(n)}$:
\ben
&& Z^{(1)} =  (\frac{t^2\epsilon^{-2}}{2}-\frac{1}{8})p_1, \\
&& Z^{(2)} =  (\frac{t^3\epsilon^{-2}}{2}-\frac{t}{8})p_2
+ ( \frac{t^4\epsilon^{-4}}{8}+ \frac{7t^2\epsilon^{-2}}{16}-\frac{15}{128}) p_1^2, \\
&& Z^{(3)} = (\frac{5\epsilon^{-2}t^4}{6} + \frac{5t^2}{12} - \frac{5\lambda^2}{32}) p_3
+ ( \frac{t^5\epsilon^{-4}}{4} + \frac{15t^3\epsilon^{-2}}{8}-\frac{31t}{64} ) p_2p_1 \\
&& \qquad + (\frac{\epsilon^{-6}t^6}{48}
+\frac{15\epsilon^{-4}t^4}{64} + \frac{419t^2\epsilon^{-2}}{768}- \frac{155}{1024}  ) p_1^3,
\een

\subsection{$\tilde{Z}_{even}$  and counting dessins}

We lose no information by taking $\epsilon =1$.
One can recover $\tilde{Z}_{even}$ from $\tilde{Z}_{even}|_{\epsilon=1}$ by making the following change of variables:
\begin{align}
t & \to \epsilon^{-1}t, & p_n &\to \epsilon^{n-1} p_n, \;\; n \geq 1.
\end{align}
Now note:
\be
W|_{\epsilon=1} = (\frac{t^2}{2} - \frac{1}{8}) p_1 + 2t \Lambda_1
+ 2 M_1,
\ee
where $\Lambda_1$ and $M_1$ are the operators introduced in \cite{Zog} and \cite{Kaz-Zog}:
\bea
&& \Lambda_1 = \sum_{i=2}^\infty (i - 1)p_i \frac{\pd}{\pd p_{i-1}}, \\
&& M_1 = \sum_{i=2}
\sum^{i-1}_{j=1} \biggl((i - 1)p_jp_{i-j} \frac{\pd}{\pd p_{i-1}}
+ j(i - j)p_{i+1} \frac{\pd^2}{\pd p_jp_{i-j}} \biggr)
\eea
Because these are exactly the same operators that we have seen in the case of
cut-and-join representation of dessin tau-function,
from
\be
\tilde{Z}_{even}|_{\epsilon=1} = e^{(\frac{t^2}{2} - \frac{1}{8}) p_1 + 2t \Lambda_1
+ 2 M_1} 1
\ee
we then get

\begin{thm}
The modified partition function of Hermitian one-matrix model with even couplings
is a special case of dessin tau-function:
\be
\tilde{Z}_{even}|_{\epsilon=1} = Z_{Dessins}|_{u= \frac{t}{2}+\frac{1}{4}, v=\frac{t}{2}-\frac{1}{4}, s= 2}.
\ee
\end{thm}

\begin{rmk}
From the point of view of combinatorics this is a very unexpected result.
It relates $\tilde{Z}_{even}|_{\epsilon=1}$ to counting all possible dessins
in a suitable way,
but from its definition, $\tilde{Z}_{even}|_{\epsilon=1}$ comes from
$Z_{even}$ which only counts clean dessins with only vertices of even valences.
\end{rmk}

As corollaries we have

\begin{thm} \label{thm:KP}
The modified partition function $\tilde{Z}_{even}$ is a one-parameter family of KP tau-functions.
\end{thm}

\begin{thm}
After the boson-fermion correspondence,
\be
\tilde{Z}_{even}|_{\epsilon=1} = e^{\tilde{A}_{even}} \vac,
\ee
where $\tilde{A}_{even}$ is some operator of the form:
\be
\tilde{A}_{even} = \sum_{m, n \geq 0} A_{m,n} \psi_{-m-1/2} \psi^*_{-n-1/2},
\ee
with the coefficients $A_{m,n}$ explicitly defined as follows:
\be
\begin{split}
A_{m,n} & = \frac{(-1)^n}{2^{3(m+n+1)} (m+n+1)\cdot m! n!} \\
& \cdot \prod_{i=0}^{2m} (2t+2i+1) \cdot \prod_{j=0}^{2n} (2t-2j-1).
\end{split}
\ee
\end{thm}

The following are the first few $A_{m,n}$,s:
\ben
&& A_{0,0} = \frac{1}{2^3} (2t+1)(2t-1), \\
&& A_{1,0} = \frac{1}{2^6 \cdot 2}(2t+5)(2t+3)(2t+1)(2t-1), \\
&& A_{0,1} = - \frac{1}{2^6\cdot 2}(2t+1)(2t-1)(2t-3)(2t-5).
\een

\section{Dessins with Only Vertices of Odd Valences and  The Generalized BGW Model}

In this Section we relate counting of dessins with only vertices of odd valences
to the generalized BGW model.

\subsection{The generalized BGW model}

The Brezin-Gross-Witten model was introduced in the lattice gauge theory \cite{BG, GW}.
It has been shown that its partition function
satisfies the Virasoro constraints \cite{Gross-Newman},
and it is a tau-function of the KdV hierarchy \cite{Mir-Mor-Sem}.
See \cite{Do-Norbury} and \cite{Norbury-Class} for topological recursions and TQFT related to
the BGW model.

Originally the BGW model is described by a unitary matrix model
\be
Z_{BGW} = \int [dU] e^{\frac{1}{\hbar} \Tr (A^\dagger U+AU^\dagger)},
\ee
but it can be also described by a generalized Kontsevich model:
\be
Z_{BGW, N}(M) ¡Ô \frac{1}{V_N} \int_{H_{N\times N}} dXe^{\Tr(MX-N \log X+\frac{1}{X} )}
\ee

\subsection{Virasoro constraints and cut-and-join for the deformed BGW tau-function}

The generalized BGW tau-function satisfies the Virasoro constraints \cite{Alex-BGW}:
\be
\hbar \hat{\cL}^{(N)}_m \tau_N(t, \hbar) = \frac{\pd}{\pd t_{2m+1}} \tau_N(t, \hbar), m \geq 0,
\ee
where
\be\begin{split}
\hat{\cL}^{(N)}_m
& =
\frac{1}{2} \sum_{k=0}^\infty (2k + 1)t_{2k+1} \frac{\pd}{\pd t_{2k+2m+1}} \\
& + \frac{1}{4} \sum_{a+b=m-1}
\frac{\pd^2}{\pd t_{2a+1}\pd t_{2b+1}}
+ (\frac{1}{16} - \frac{N^2}{4}) \delta_{m,0}.
\end{split}
\ee
It follows that $Z_{BGW}^{(N)}(t)$ has the following cut-and-join
representation:
\be
Z_{BGW}^{(N)}(t) = e^{\hbar W_{BGW}^{(N)}}  (1)
\ee
where the operator $W_{BGW}^{(N)}$ is defined by:
\be
\begin{split}
& W_{BGW}^{(N)}  = \frac{1}{2} \sum_{k,m=0}^\infty
(2k + 1)(2m + 1)t_{2k+1}t_{2m+1} \frac{\pd}{\pd t_{2k+2m+1}} \\
&  +  \frac{1}{4} \sum_{k,m=0}^\infty
(2k + 2m + 3)t_{2k+2m+3} \frac{\pd^2}{\pd t_{2k+1} \pd t_{2m+1}}
+ \biggl( \frac{1}{16} - \frac{N^2}{4} \biggr) t_1.
\end{split}
\ee

\subsection{Relating generalized BGW to counting dessins of odd valence}

Recall by \eqref{eqn:Z-CJ}£¬
the dessin tau-function has the following cut-and-join representation:
\be
Z_{Dessins}(u,v,s) = e^{W_{Dessins}}1,
\ee
where $W_{Dessins}$ is the following differential operator:
\ben
&& W_{Dessins} = s(u+v)\sum_{i=2}^\infty (i - 1)p_i \frac{\pd}{\pd p_{i-1}} \\
&& + s \sum_{i=2} \sum^{i-1}_{j=1}
\biggl((i - 1)p_jp_{i-j} \frac{\pd}{\pd p_{i-1}}
+ j(i - j)p_{i+1} \frac{\pd^2}{\pd p_j \pd p_{i-j}} \biggr)
+ suvp_1.
\een
So comparing $\hbar W_{BGW}^{(N)}$ with $W_{Dessins}$,
we get:

\begin{prop}
The cut-and-join operator for the generalized BGW tau-function is a special case of
the cut-and-join operator for the dessin tau-function:
\be
\hbar W_{BGW}^{(N)} = W_{Dessins}|_{p_{2n}=0, p_{2n-1}= \sqrt{2} (2n-1) t_{2n-1}, s =\frac{\hbar}{2\sqrt{2}},
u= -v=  \sqrt{\frac{N^2}{2}-\frac{1}{8}} }.
\ee
\end{prop}

\subsection{Kac-Schwarz operators for generalized BGW theory}

As proved in \cite{Alex-BGW},
from the general theory of GKM it follows that the generalized BGW theory has a phase that
 defines a tau-function of the MKP hierarchy,
specified by the basis vector:
\be \label{eqn:Phi-N-j}
\Phi^{(N)}_j (z) :
 = z^{j-1}\biggl(1   + \sum_{k=1}^\infty
\frac{(-\hbar)^k}{z^k} \frac{a_k(j - N)}{16^k k!} \biggr),
\ee
where $a_k(j)$ is a polynomial given below:
\be
a_k(j) = \prod_{l=1}^{k} (2(j-1)-(2l-1))(2(j-1)+(2l-1)).
\ee
One can also rewrite it as follows:
\be
a_k(j) = 2^{2k}[j-\half]_{-k}^{k-1}.
\ee
The following Kac-Schwarz operators are found in \cite{Alex-BGW}:
\bea
&& a = \frac{z}{2} \frac{\pd}{\pd z} + \frac{z}{\hbar} - \frac{1}{4}, \\
&& b = z^2, \\
&& c_N = \frac{1}{b} \biggl(a^2-\frac{N^2}{4}\biggr).
\eea
The following recursion relations are satisfied:
\bea
&& a \Phi^{(N)}_j  = \biggl(j - 1 - \frac{N}{2} \biggr)  \Phi_j^{(N)} + \frac{1}{\hbar} \Phi_{j+1}^{(N)},\\
&& b \Phi^{(N)}_j = (j - N)\hbar \Phi_{j+1}^{(N)} + \Phi_{j+2}^{(N)}, \\
&& c_N \Phi_j^{(N)} = \frac{1}{\hbar} (j - 1)\Phi_{j-1}^{(N)} + \frac{1}{\hbar^2} \Phi^{(N)}_j.
\eea

\subsection{Affine coordinates for generalized BGW theory}

In order to get the affine coordinates corresponding to $\tau_{BGW}^{(N)}$,
we now convert $\{ \Phi_j^{(N)} \}_{j\geq 0}$ into the normalized basis:
\ben
\Psi^{(N)}_0(z) & = & \Phi_1^{(N)}(z),  \\
\Psi^{(N)}_1(z) & = & \Phi_2^{(N)}(z) + \frac{a_1(2-N)}{16} \hbar\Phi_1^{(N)}(z), \\
\Psi^{(N)}_2(z) & = & \Phi_3^{(N)}(z) + \frac{a_1(3-N)}{16} \hbar\Phi_2^{(N)}(z) \\
& - & \hbar^2\biggl( \frac{a_2(3-N)}{16^2 \cdot 2!}- \frac{a_1(3-N)}{16^1\cdot 1!} \frac{a_1(2-N)}{16^1\cdot 1!}
\biggr) \Phi^{(N)}_1(z),
\een
etc.,
\ben
\Psi^{(N)}_1(z) & = & \Phi_2^{(N)}(z) + \frac{\hbar(2N-1)(2N-3)}{16} \Phi_1^{(N)}(z), \\
\Psi^{(N)}_2(z) & = & \Phi_3^{(N)}(z) + \frac{\hbar(2N-3)(2N-5)}{16} \Phi_2^{(N)}(z) \\
& + & \frac{\hbar^2(2N+1)(2N-1)(2N-3)(2N-5)}{16^22!} \Phi_1^{(N)}(z), \\
\Psi^{(N)}_3(z) & = & \Phi_4^{(N)}(z) + \frac{\hbar(2N-5)(2N-7)}{16} \Phi_3^{(N)}(z) \\
& + & \frac{\hbar^2(2N-1)(2N-3)(2N-5)(2N-7)}{16^22!} \Phi_2^{(N)}(z) \\
& + & \frac{\hbar^3(2N+3)(2N+1)\cdots (2N-7)}{16^33!} \Phi_1^{(N)}(z).
\een
In general we have:

\begin{prop}
For the generalized BGW theory,
the normalized basis $\{\Psi_k^{(N)}\}_{k \geq 0}$
is related to the basis $\{\Phi_j^{(N)}\}_{j \geq 1}$ as follows:
\be
\Psi^{(N)}_{m-1}(z)
= \Phi^{(N)}_m(z)  + \sum_{k=1}^{m-1} \frac{\hbar^k}{16^kk!}
\prod_{j=1}^{2k} (2N-2m-1 + 2j) \cdot
\Phi_{m-k}^{(N)}(z).
\ee
\end{prop}

\begin{proof}
For $m \geq 2$, we have
\ben
&&\Psi^{(N)}_{m-1}(z)= \Phi^{(N)}_m(z)
+ \sum_{k=1}^{m-1} c_k \Phi_{m-k}^{(N)}(z) \\
& = & z^{m-1} + \sum_{l=1}^\infty (-\hbar)^l \frac{a_l(m - N)}{16^l l!} z^{m-1-l} \\
&  + & \sum_{k=1}^{m-1} c_k \biggl(
z^{m-k-1} + \sum_{l=1}^\infty (-\hbar)^l \frac{a_l(m -k- N)}{16^l l!} z^{m-k-1-l} \biggr).
\een
By comparing the coefficients of $z^{m-2}, z^{m-3}$, $\dots$, $z^0$,
we get the following recursions relations:
\ben
&& c_1- \hbar \cdot \frac{a_1(m-N)}{16} = 0, \\
&& c_k -c_{k-1} \cdot \hbar \frac{a_1(m-k+1-N)}{16^11!}
+ c_{k-2} \cdot \hbar^2\frac{a_2(m-k+2-N)}{16^22!} \\
&& \qquad + \cdots
+ (-\hbar)^k \frac{a_k(m-N)}{16^kk!} = 0.
\een
\ben
&& c_1 = \frac{\hbar}{4} [N-m-1/2]_1^2, \\
&& c_k = c_{k-1} \cdot \hbar \frac{[N-m-1/2]_k^{k+1}}{4^11!}
- c_{k-2} \cdot \hbar^2\frac{[N-m-1/2]_{k-2}^{k+1}}{4^22!} \\
&& \qquad + \cdots
+ (-1)^{k-1} \hbar^k \frac{[N-m-1]^{k+1}_{-(k-2)}}{4^kk!}.
\een
One can prove that
\be
c_k = \frac{\hbar^k}{16^kk!}
\prod_{j=1}^{2k} (2N-2m-1 + 2j)
= \frac{\hbar^k}{4^kk!} [N-m-1/2]_1^{2k}.
\ee
satisfy these recursion relations.
The proof can be reduced to  the following well-known higher order difference identity:
\be
\sum_{j=0}^m (-1)^j \binom{m}{j} \prod_{k=0}^{m-2} (x+k ) = 0.
\ee
\end{proof}

As a corollary we get the following:

\begin{thm}
We have the following explicit  expression for
the normalized basis hence the affine coordinates
for the generalized BGW tau-function:
\be
\begin{split}
\Psi^{(N)}_{m-1}
= & z^{m-1}
+ \sum_{n=1}^\infty z^{-n}
\sum_{k=0}^{m-1} \frac{(-1)^{m+n-1-k}\hbar^{m+n-1}}{4^{m+n-1}k!(m+n-1-k)!} \\
& \cdot [N-m-1/2]_{1}^{2k}[N-n-1/2]_{2k-2m+3}^0.
\end{split}
\ee
\end{thm}

The geometric interpretation of the BGW theory has been given by Norbury classes \cite{Norbury-Class}.
We will return to a discussion of them in a later part of the series.

\subsection{Quantum spectral curve for generalized BGW theory}

This has been discussed in \cite{Alex-BGW}.
It is related to the modified Bessel equation by taking $x= z^2$.
Here we write down the equation of hypergeometric
type satisfied by the principal specialization.
By \eqref{eqn:Phi-N-j}, we have
\be
\Phi^{(N)}_1 (z)
 =  1   + \sum_{k=1}^\infty
\frac{\hbar^k}{z^k} \frac{\prod_{l=1}^k (2l-1)^2}{16^k k!}.
\ee
 Hence the principal specialization of the generalized BGW theory is:
\be
\psi^{(N)}  (t)
 =  1   + \sum_{k=1}^\infty \hbar^kt^k \frac{\prod_{l=1}^k (2l-1)^2}{16^k k!}.
\ee
It satisfies the following equation of hypergeometric type:
\be
\frac{d}{dt} \psi = \frac{\hbar}{16}(2t\frac{d}{dt}+1)^2 \psi.
\ee

\section{Conclusions an Prospects}

In the context of this paper we have provided some concrete examples to illustrate our proposal 
to study dualities in the framework of moduli space of theories via the the theory of KP hierarchy.  
Tau-functions, Virasoro constraints, cut-and-join representations,
Kac-Schwarz operators, fermionic representations, 
quantum spectral curves and hypergeometric equations, 
etc.,
together with spectral curves, Eynard-Orantin topological recursion,
and the more recently developed notion of emergent geometry \cite{Zhou-Emerg},
supply a plethora of tools to examine each individual theory and help 
to make connections among  them.

The enumeration problems of Grothendieck's dessins d'enfants 
have been demonstrated to play a central role among those theories 
that are related to 2D topological gravity or minimal matters coupled with 2D gravity.
Simply by comparing the cut-and-join representations,
we  see that some other enumeration problems arising in matrix models are special cases.

The idea of counting of fat graphs, 
introduced by t' Hooft \cite{tHt},
have been developed in the physics literature in at least two directions.
In one direction,
through large-N expansions of matrix models,
a connection between the theories of minimal matters coupled with 2D topological gravity
with theory of integrable hierarchy has been built.
On the other hand,
large-N expansions of gauge theories on various dimensions lead to gauge/geometry duality.
For example, 
2D Yang-Mills theory leads to Hurwitz numbers via Gross-Taylor expansion,
3D Chern-Simons theory \cite{Witten-CS} leads to the theory of topological vertex
for Gromov-Witten theory of toric Calabi-Yau 3-folds via \cite{Witten-TS-CS, Gop-Vaf, Oog-Vaf, Mar-Vaf, AMV, AKMV},
4D Yang-Mills theory leads to AdS/CFT correspondence \cite{Mal}.
If we abbreviate the trility of gauge/gravity/geometry as 3G,
and abbreviate the quadrality obtained from it by adding groups as 4G,
then it is natural to add another item, the Grassmannian, to form a quintility,
and abbreviate it as 5G:
\ben
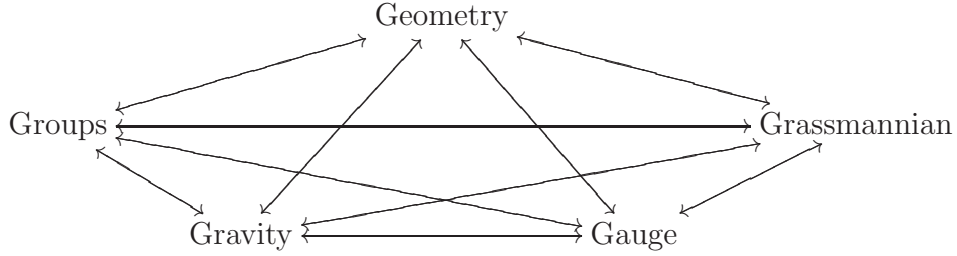
\xymatrix{
& & \text{Geometry} \ar[lld] \ar[rrd] \ar[ldd] \ar[rdd] &&  \\ 
\text{Groups} \ar[rd] \ar[rru] \ar[rrrr] \ar[rrrd]
& & && \text{Grassmannian} \ar[llll] \ar[llu] \ar[ld] \ar[llld] \\
&\text{Gravity} \ar[rr] \ar[lu] \ar[ruu] \ar[rrru] & &
\text{Gauge} \ar[ll] \ar[ru] \ar[luu] \ar[lllu] &}
\een

We note that the cut-and-join equations have played a crucial role in
the development of the mathematical theory of topological vertex \cite{Zhou-HH, LLZ1, LLZ2, LLLZ}.
This is a hint that the enumerative aspect of Grothendieck's theory of dessins d'enfants
can be unified with the study of mirror symmetry.
If we add Grothendieck to  the above 5G network
then we reach the level of 6G. 
We will report such unification in later parts of the series.

\vspace{.5in}
{\em Acknowledgements}.
The author thanks Professor Motohico Mulase for introducing him to Grothendieck's theory
of dessins d'enfants about ten years ago.
This research is partially supported by   NSFC grants 11661131005 and 11890662.

\begin{appendix}

\section{Proofs of Theorem \ref{thm:Main1} and Theorem \ref{thm:Main2}}

In this Appendix
we will prove Theorem \ref{thm:Main2} by converting the Virasoro constraint in the fermionic
picture.
Our strategy is the same as in \cite{Zhou-WK}: We convert the Virasoro constraints into
the fermionic picture.
We also derive  Theorem \ref{thm:Main1} from Theorem \ref{thm:Main2}.

\subsection{Bosonic reformulation of the Virasoro constraints
for dessin tau-function}

Consider a field:
\ben
\beta(z) & = &  \frac{(u +v)}{2z}
+ \half \sum_{j=1}^\infty (p_j-\frac{\delta_{j,1}}{s}) z^{j-1}
+ \sum_{i=1}^\infty i z^{-i-1}\frac{\pd}{\pd p_i},
\een
we rewrite it in the following form:
$\beta(z) = \sum_{n \in \bZ} \beta_n z^{-n-1}$, where
\be
\beta_n = \begin{cases}
  n \frac{\pd}{\pd p_n}, & n > 0, \\
\frac{u+v}{2}, & n =0 , \\
(p_{-n} -\frac{\delta_{-n,1}}{s}), & n < 0.
\end{cases}
\ee
We take $\beta_n$ with $n \leq 0$ to be creators
and $\beta_n$ with $n > 0$ to be annihilators,
and with these choice define the normally ordered
products $:\beta_m\beta_n$ and
\ben
L(z):&= &:\beta(z)\beta(z): - \frac{(u-v)^2}{4z^2} \\
& = & \frac{(u+v)^2}{4z^2} + \frac{1}{4} \biggl( \sum_{j=1}^\infty (p_j-\frac{\delta_{j,1}}{s}) z^{j-1} \biggr)^2
+ \biggl( \sum_{i=1}^\infty i z^{-i-1}\frac{\pd}{\pd p_i} \biggr)^2 \\
& + & \frac{u+v}{2z}\sum_{j=1}^\infty (p_j-\frac{\delta_{j,1}}{s}) z^{j-1}
+ \frac{u+v}{2z} \sum_{i=1}^\infty i z^{-i-1}\frac{\pd}{\pd p_i} \\
& + & \sum_{j=1}^\infty (p_j-\frac{\delta_{j,1}}{s}) z^{j-1}
\sum_{i=1}^\infty i z^{-i-1}\frac{\pd}{\pd p_i} - \frac{(u-v)^2}{4z^2}.
\een
Write
\be
L(z) = \sum_{n \in \bZ} L_n z^{-n-2}.
\ee
Then $\{L_n\}_{n \geq 0}$ are exactly the Virasoro constraints
for $Z_{Dessins}$.

\subsection{Fermionic formulation of the Virasoro operator $L_0$}

As in \cite{Zhou-WK},
consider the fields of fermionic operators:
\ben
&& \psi(z) = \sum_{r\in \bZ + \frac{1}{2}} \psi_r z^{-r-\frac{1}{2}}, \\
&& \psi^*(z) = \sum_{s\in \bZ + \frac{1}{2}} \psi_r^* z^{-s-\frac{1}{2}},
\een
The commutation relations
\be
[\psi_r, \psi_s]_+ =0, \quad [\psi_r^*, \psi_s^*] = 0, \quad
[\psi_r, \psi^*_s] = \delta_{r+s,0}
\ee
are equivalent to the following OPE's:
\ben
&& \psi(z) \psi^*(w) \sim :\psi(z)\psi^*(w): + \frac{1}{z-w}, \\
&& \psi^*(z) \psi(w) \sim :\psi^*(z)\psi(w): + \frac{1}{z-w},
\een
Write
\be
\gamma(z) = :\psi(z)\psi^*(z): = \sum_{n \in \bZ} \gamma_n z^{-n-1},
\ee
where
\be
\gamma_n = \sum_{\substack{r,s\in \bZ+1/2\\r+s=n }} :\psi_r\psi_s^*:.
\ee
By boson-fermion correspondence,
one identifies $\gamma_n$ with $\beta_n$ for $n \neq 0$.
By Wick's theorem,
\ben
\gamma(z)\gamma(w)
& = & \frac{1}{(z-w)^2} + :\pd_w \psi(w)\psi^*(w):  + :\pd_w \psi^*(w)\psi(w): + \cdots.
\een
It follows that
\be \label{eqn:bos-fer}
\sum_{a+b=n} :\gamma_{a}\gamma_{b}:
=  \sum_{r+s=n} (s-r) :\psi_r\psi_s^*: .
\ee
So we have
\ben
L_0 & = &  -\frac{1}{s} \frac{\pd}{\pd p_{1}}
+ \sum_{j=1}^\infty p_j \cdot j \frac{\pd}{\pd p_{j}}
+ uv \\
& = &  -\frac{1}{s} \gamma_1
+ \half \sum_{a+b=0} :\gamma_a\gamma_b: + uv \\
& = & - \frac{1}{s}  \sum_{\substack{r_1,s_1\in \bZ+1/2\\r_1+s_1= 1 }} :\psi_{r_1}\psi_{s_1}^*:
+ \half \sum_{r_2+s_2=0} (s_2-r_2) :\psi_{r_2}\psi_{s_2}^*:
+ uv.
\een

\subsection{Recursions for $A_{m,n}$ from $L_0$}

Now we have
\ben
L_{0} & = & - \frac{1}{s}
\sum_{r_1+s_1= 1}
 :\psi_{r_1}\psi_{s_1}^*: + \frac{1}{2} \sum_{r_2+s_2=0} (s_2-r_2) :\psi_{r_2}\psi_{s_2}^*: + uv\\
& = & uv - \frac{1}{s} \biggl(\psi_{\frac{1}{2}} \psi_{\frac{1}{2}}^*
+ \sum_{k=0}^\infty (\psi_{-k-\frac{1}{2}} \psi^*_{k+\frac{3}{2}}
- \psi^*_{-k-\frac{1}{2}} \psi_{k+\frac{3}{2}} ) \biggr)\\
& + &  \frac{1}{2}
\sum_{l=0}^\infty (2l+1) (\psi_{-l-\frac{1}{2}} \psi^*_{l+\frac{1}{2}}
+ \psi^*_{-l-\frac{1}{2}} \psi_{l+\frac{1}{2}} ).
\een

\ben
e^{-A} L_{0} Z_{Dessin}  & = & \biggl[ uv - \frac{A_{0,0}}{s} - \frac{1}{s}\biggl(-\sum_{m, n \geq 0} A_{m,0} A_{0,n} \psi_{-m-\frac{1}{2}} \psi^*_{-n-\frac{1}{2}} \\
& + &\sum_{k,l \geq 0} ( \psi_{-k-\frac{1}{2}} A_{k+1,l} \psi^*_{-l-\frac{1}{2}}
+ \psi^*_{-k-\frac{1}{2}} A_{l,k+1} \psi_{-l-\frac{1}{2}} ) \biggr)\\
& + &  \frac{1}{2}  \sum_{k,l \geq 0} (2l+1) (\psi_{-l-\frac{1}{2}} A_{l,k} \psi^*_{-k-\frac{1}{2}}
- \psi^*_{-l-\frac{1}{2}} A_{k,l} \psi_{-k-\frac{1}{2}} ) \biggr] \vac.
\een
By comparing the coefficients of $\vac$ and $\psi_{-m-\frac{1}{2}}\psi^*_{-n-\frac{1}{2}}\vac$,
we get the following  equations:
\be
A_{0,0} = suv,
\ee
and for $m > 0$ and $n \geq 0$ or $ m\geq 0$ and $n > 0$,
\be
\begin{split}
&  - \frac{1}{s} \big(- A_{m,0} A_{0,n} + A_{m+1,n} - A_{m,n+1} \big) \\
+ &  \frac{1}{2} \biggl( (2m+1) A_{m,n} + (2n+1) A_{m,n}     \biggr) = 0.
\end{split}
\ee
This can be converted into the following  recursion relations:
\be \label{eqn:L-1Recursion}
A_{m,n+1} =   A_{m+1,n}  - A_{m,0} A_{0,n}   - s(m+n+1) A_{m,n}.
\ee
This will uniquely determine all $A_{m,n}$'s with $n > 0$
by all $A_{m,0}$'s.
The following are the first few examples:
\ben
A_{0,1} & = & A_{1,0} - A_{0,0} A_{0,0} - s A_{0,0}, \\
A_{1,1} & = & A_{2,0} - A_{1,0} A_{0,0} -2s A_{1,0}, \\
A_{0,2} & = & A_{1,1} - A_{0,0} A_{0,1} -2s A_{0,1}, \\
A_{2,1} & = & A_{3,0} - A_{2,0} A_{0,0} -3s A_{2,0}, \\
A_{1,2} & = & A_{2,1} - A_{1,0} A_{0,1} -3s A_{1,1}, \\
A_{0,3} & = & A_{1,2} - A_{0,0} A_{0,2} -3s A_{0,2}.
\een

\subsection{Recursions for $A_{m,n}$ from $L_1$}

Similarly,
\ben
L_1 & = & -\frac{2}{s} \frac{\pd}{\pd p_{2}}
+(u +v)  \frac{\pd}{\pd p_1}
+ \sum_{j=1}^\infty p_j(n + 1) \frac{\pd}{\pd p_{n+ j}} \\
& = & - \frac{1}{s} \gamma_2 + (u+v) \gamma_1
+ \half \sum_{a+b=1} :\gamma_a\gamma_b: \\
& = & - \frac{1}{s}  \sum_{\substack{r_1,s_1\in \bZ+1/2\\r_1+s_1= 2 }} :\psi_{r_1}\psi_{s_1}^*:
+ (u+v) \sum_{\substack{r_2,s_2\in \bZ+1/2\\r_2+s_2= 1 }} :\psi_{r_2}\psi_{s_2}^*: \\
& + & \half \sum_{\substack{r_3,s_3\in \bZ+1/2\\r_3+s_3= 1 }} (s_3-r_3) :\psi_{r_3}\psi_{s_3}^*:,
\een
and so  we have
\ben
L_{1}
& = &  - \frac{1}{s} \biggl(\psi_{\frac{3}{2}} \psi_{\frac{1}{2}}^*+ \psi_{\frac{1}{2}} \psi_{\frac{3}{2}}^*
+ \sum_{k=0}^\infty (\psi_{-k-\frac{1}{2}} \psi^*_{k+\frac{5}{2}}
- \psi^*_{-k-\frac{1}{2}} \psi_{k+\frac{5}{2}} ) \biggr)\\
& + & (u+v) \biggl(\psi_{\frac{1}{2}} \psi_{\frac{1}{2}}^*
+ \sum_{k=0}^\infty (\psi_{-k-\frac{1}{2}} \psi^*_{k+\frac{3}{2}}
- \psi^*_{-k-\frac{1}{2}} \psi_{k+\frac{3}{2}} ) \biggr) \\
& + &  \sum_{l=0}^\infty (l+1) (\psi_{-l-\frac{1}{2}} \psi^*_{l+\frac{3}{2}}
+ \psi^*_{-l-\frac{1}{2}} \psi_{l+\frac{3}{2}} ),
\een
and therefore,
\ben
e^{-A} L_{1} Z_{Dessins}
& = & \biggl[- \frac{A_{0,1}}{s}
+ \frac{1}{s}\sum_{m, n \geq 0} A_{m,1} A_{0,n} \psi_{-m-\frac{1}{2}} \psi^*_{-n-\frac{1}{2}} \\
& - & \frac{A_{1,0}}{s}
+ \frac{1}{s} \sum_{m, n \geq 0} A_{m,0} A_{1,n} \psi_{-m-\frac{1}{2}} \psi^*_{-n-\frac{1}{2}} \\
& - & \frac{1}{s} \sum_{k,l \geq 0} ( \psi_{-k-\frac{1}{2}} A_{k+2,l} \psi^*_{-l-\frac{1}{2}}
+ \psi^*_{-k-\frac{1}{2}} A_{l,k+2} \psi_{-l-\frac{1}{2}} )  \\
& + & (u+v) \biggl(A_{0,0}  - \sum_{m, n \geq 0} A_{m,0} A_{0,n}
\psi_{-m-\frac{1}{2}} \psi^*_{-n-\frac{1}{2}}  \biggr) \\
& + & (u+v)\sum_{k,l \geq 0} ( \psi_{-k-\frac{1}{2}} A_{k+1,l} \psi^*_{-l-\frac{1}{2}}
+ \psi^*_{-k-\frac{1}{2}} A_{l,k+1} \psi_{-l-\frac{1}{2}} )  \\
& + &  \sum_{k,l \geq 0} (l+1) (\psi_{-l-\frac{1}{2}} A_{l+1,k} \psi^*_{-k-\frac{1}{2}}
- \psi^*_{-l-\frac{1}{2}} A_{k,l+1} \psi_{-k-\frac{1}{2}} ) \biggr] \vac.
\een
By comparing the coefficients of $\vac$ and $\psi_{-m-\frac{1}{2}}\psi^*_{-n-\frac{1}{2}}\vac$,
we get the following  equations:
\be \label{eqn:A10+A01}
- \frac{A_{0,1}}{s} - \frac{A_{1,0}}{s} + (u+v) A_{0,0}  = 0,
\ee
and for $m > 0$ and $n \geq 0$ or $ m\geq 0$ and $n > 0$,
\be
\begin{split}
&  - \frac{1}{s} \big(- A_{m,1} A_{0,n} - A_{m,0}A_{1,n}
+ A_{m+2,n} - A_{m,n+2} \big) \\
+ & (u+v) \big( -A_{m,0} A_{0,n} + A_{m+1, n} - A_{m, n+1} \big) \\
+ &  \biggl( (m+1) A_{m+1,n} + (n+1) A_{m,n+1}     \biggr) = 0.
\end{split}
\ee
This can be converted into the following  recursion relations:
\be
\begin{split}
A_{m,n+2} = &  A_{m+2,n} - A_{m,1} A_{0,n} - A_{m,0}A_{1,n}  \\
- & s(u+v) \big(-A_{m,0} A_{0,n} + A_{m+1, n} - A_{m, n+1} \big) \\
- &  s \biggl( (m+1) A_{m+1,n} + (n+1) A_{m,n+1}     \biggr).
\end{split}
\ee
The following are the first few examples:
\ben
A_{0,2} & = &  A_{2,0} - A_{0,1} A_{0,0} - A_{0,0}A_{1,0} \\
& - & s(u+v) \big(-A_{0,0} A_{0,0} + A_{1, 0} - A_{0, 1} \big)
- s \biggl( A_{1,0} + A_{0,1}     \biggr), \\
A_{1,2} & = &  A_{3,0} - A_{1,1} A_{0,0} - A_{1,0}A_{1,0}  \\
& - & s(u+v) \big( -A_{1,0} A_{0,0} + A_{2, 0} - A_{1, 1} \big)
- s \biggl( 2A_{2,0} + A_{1,1}     \biggr), \\
A_{0,3} & = &  A_{2,1} - A_{0,1} A_{0,1} - A_{0,0}A_{1,1}  \\
& - & s(u+v) \big( -A_{0,0} A_{0,1} + A_{1, 1} - A_{0, 2} \big)
- s \biggl( A_{1,1} + 2 A_{0,2}     \biggr).
\een
From these concrete examples,
we see that we can determine $A_{1,0}$ and $A_{0,1}$ from
\eqref{eqn:A10+A01} from the $L_1$-constraint,
and the relations for $A_{1,0}$ and $A_{0,1}$ from
the $L_0$-constraint.
However the $L_0$-constraint and the $L_1$-constraint
are not sufficient to determine $A_{m,0}$ for $m > 1$.

\subsection{Relations for $A_{k,l}$ from the $L_n$-constraints
for $n > 1$}
Recall for $n > 1$,
$L_n$ is given by:
\ben
L_n & = & -\frac{n + 1}{s} \frac{\pd}{\pd p_{n+1}}
+(u +v)n \frac{\pd}{\pd p_n}
+ \sum_{j=1}^\infty p_j(n + j) \frac{\pd}{\pd p_{n+ j}} \\
& + & \sum_{i + j=n} i j \frac{\pd^2}{\pd p_i\pd p_j} \\
& = & - \frac{1}{s} \gamma_{n+1}
+ (u+v) \gamma_n + \sum_{j=1} \gamma_{-j} \gamma_{n+j}
+ \sum_{i+j=n, i,j\geq 1} \gamma_i \gamma_j.
\een
Note we have
\be
\sum_{i+j=n} :\gamma_i\gamma_j:
= 2 \sum_{j=1} \gamma_{-j} \gamma_{n+j}
+ \sum_{i+j=n, i,j\geq 1} \gamma_i \gamma_j,
\ee
so we cannot directly apply \eqref{eqn:bos-fer}
to rewrite $L_n$ in the fermionic picture.
Nevertheless,
we can apply the following trick.
If we assign $\deg p_j = j$
then we can decompose  everything into homogeneous parts
with respect to this grading.
For a formal power series $X$ in $p_j$'s,
denote by $(X)_m$ the homogeneous part with degree $m$.
Consider the degree zero part of the
action of $L_n$ on $(Z_{Dessins})_{n+1}$:
\ben
0 & = & [L_n (Z_{Dessin})_{n+1}]_0 \\
& = & \biggl[\biggl( - \frac{1}{s} \gamma_{n+1}
+ (u+v) \gamma_n + \sum_{j=1} \gamma_{-j} \gamma_{n+j}
+ \sum_{i+j=n, i,j\geq 1} \gamma_i \gamma_j\biggr)
(Z_{Dessin})_{n+1} \biggr]_0 \\
& = & \biggl[\biggl(- \frac{1}{s} \gamma_{n+1}+ (u+v) \gamma_n + 2 \sum_{j=1} \gamma_{-j} \gamma_{n+j}
+ \sum_{i+j=n, i,j\geq 1} \gamma_i \gamma_j\biggr)
(Z_{Dessin})_{n+1} \biggr]_0 \\
& = & \biggl[\biggl(- \frac{1}{s} \gamma_{n+1}+ (u+v) \gamma_n
+ \sum_{i+j=n} :\gamma_i \gamma_j: \biggr)
(Z_{Dessin})_{n+1} \biggr]_0 \\
& = & \biggl[\biggl(- \frac{1}{s}
\sum_{\substack{r_1,s_1\in \bZ+1/2\\r_1+s_1= n+1 }} :\psi_{r_1}\psi_{s_1}^*:
+ (u+v) \sum_{\substack{r_2,s_2\in \bZ+1/2\\r_2+s_2= n }} :\psi_{r_2}\psi_{s_2}^*: \\
& + & \half \sum_{\substack{r_3,s_3\in \bZ+1/2\\r_3+s_3= n}} (s_3-r_3) :\psi_{r_3}\psi_{s_3}^*:
\biggr) (Z_{Dessin})_{n+1} \biggr]_0.
\een
This gives us
\be
\sum_{j=0}^n A_{n-j,j}
=s (u+v) \sum_{j=0}^{n-1} A_{n-1-j,j}
+ s \sum_{j=0}^{n-1} (n-1-2j) \cdot A_{n-1-j,j}.
\ee

\begin{prop}
For $Z_{Dessins}$ given by \eqref{eqn:Z-A},
the coefficients $A_{m,n}$ are determined by the following recursion relations and initial value:
\bea
&& A_{0,0} = suv, \label{eqn:A00} \\
&& A_{m,n+1} =   A_{m+1,n}  - A_{m,0} A_{0,n}   - s(m+n+1) A_{m,n},  \label{eqn:Rec1} \\
&& \sum_{j=0}^n A_{n-j,j}
=s (u+v) \sum_{j=0}^{n-1} A_{n-1-j,j}
+ s \sum_{j=0}^{n-1} (n-1-2j) \cdot A_{n-1-j,j}. \label{eqn:Rec2}
\eea
\end{prop}

\begin{proof}
In the above we have shown that these equations are satisfied by $A_{m,n}$.
Now we show that they suffices to inductively determine them.
For $n > 0$ we get from \eqref{eqn:Rec1} an expression of the form
\be
A_{0,n} = A_{n,0} + \sum_{i+j<n} a_{i,j} \cdot A_{i,j}
\ee
for some coefficients $a_{i,j}$,
and similarly from \eqref{eqn:Rec2} we get an expression of the form:
\be
A_{0,n} = - A_{n,0} + \sum_{i+j<n} b_{i,j} \cdot A_{i,j}
\ee
for some coefficients $b_{i,j}$.
Putting these together we get a recursion of the form:
\be
A_{0,n} =  \sum_{i+j<n} a_{i,j} \cdot A_{i,j}
\ee
for some coefficients $c_{i,j}$,
hence $A_{0,n}$ and $A_{n,0}$ can be determined from $\{A_{i,j}\}_{i+j<n}$.
One can then use \eqref{eqn:Rec2} again to determine all $A_{i,j}$'s with $i+j=n$.
\end{proof}

\subsection{Proof of the explicit formula for $A_{m,n}$}

In this subsection we will prove \eqref{eqn:Amn}
which we recall here:
\ben
A_{m, n}
= \frac{(-1)^n s^{m+n+1}uv}{(m+n+1)m!n!} \prod_{j=1}^{m} (u+j)(v+j) \cdot \prod_{i=1}^n (u-i)(v-i).
\een
It is clear that \eqref{eqn:A00} is satisfied.
We  next verify \eqref{eqn:Rec1}:
\ben
&& A_{m+1, n} - A_{m,n+1} \\
& = &  \frac{(-1)^n s^{m+n+2}uv}{(m+n+2)(m+1)!n!} \prod_{j=1}^{m+1} (u+j)(v+j) \cdot \prod_{i=1}^n (u-i)(v-i) \\
& - & \frac{(-1)^{n+1} s^{m+n+2}}{(m+n+2)m!(n+1)!} uv \prod_{j=1}^{m} (u+j)(v+j) \cdot \prod_{i=1}^{n+1} (u-i)(v-i) \\
& = & \frac{(-1)^n s^{m+n+2} uv}{(m+n+2)(m+1)!(n+1)!}
\prod_{j=1}^{m} (u+j)(v+j) \cdot \prod_{i=1}^n (u-i)(v-i) \\
&& \cdot [(n+1)(u+m+1)(v+m+1) + (m+1)(u-n-1)(v-n-1)] \\
& = & \frac{(-1)^n s^{m+n+2} uv}{(m+n+2)(m+1)!(n+1)!}
\prod_{j=1}^{m} (u+j)(v+j) \cdot \prod_{i=1}^n (u-i)(v-i) \\
&& \cdot (m+n+2)[uv+(m+1)(n+1)] \\
& = & A_{m,0} A_{0,n} + s(m+n+1) A_{m,n}.
\een
Next we verify \eqref{eqn:Rec2}.
\ben
&& \sum_{k=0}^{m-1} \biggl(s(u+v+ m-1-2k) \cdot A_{m-1-k,k}  - A_{m-k,k} \biggr) \\
& = &  \sum_{k=0}^{m-1} \biggl[ s(u+v+ m-1-2k) \\
&& \cdot \frac{(-1)^k s^muv}{m \cdot (m-1-k)!k!} \prod_{j=1}^{m-1-k} (u+j)(v+j) \cdot \prod_{i=1}^k (u-i)(v-i) \\
& - &  \frac{(-1)^k s^{m+1}uv}{(m+1)(m-k)!k!} \prod_{j=1}^{m-k} (u+j)(v+j) \cdot \prod_{i=1}^k (u-i)(v-i) \biggr] \\
& = & \sum_{k=0}^{m-1} \frac{(-1)^{k+1} s^{m+1}uv}{m(m+1)(m-k)!k!} \prod_{j=1}^{m-k-1} (u+j)(v+j) \cdot \prod_{i=1}^k (u-i)(v-i) \\
&& \biggl(m (u+m-k)(v+m-k)
- (m+1) (m-k) (u+v+m-1-2k) \biggr) \\
& = & \sum_{k=0}^{m-1}\frac{(-1)^{k+1} s^{m+1}uv}{m(m+1)(m-k)!k!}
\prod_{j=1}^{m-k-1} (u+j)(v+j) \cdot \prod_{i=1}^k (u-i)(v-i) \\
&& \biggl(m uv - (m-k)u - (m-k) v  + (m-k) (mk+2k+1) \biggr) \\
& = & \frac{(-1)^{m} s^{m+1}uv}{(m+1)!}
\prod_{i=1}^m (u-i)(v-i)  \\
& = & A_{0,m}.
\een
This proves \eqref{eqn:Rec2}.
Hence Theorem \ref{thm:Main2} is proved.

\subsection{Proof of Theorem \ref{thm:Main1}}

By \eqref{eqn:Z-A},
in the fermionic picture we have
\be\begin{split}
& Z_{dessins} \\
= & \sum_{(m_1, \dots, m_k|n_1, \dots,n_k)}
\det (A_{m_i, n_j})_{1 \leq i, j \leq k} \cdot
\prod_{j=1}^ k\psi_{-m_j-1/2}\psi^*_{-n_j-1/2} \vac,
\end{split}
\ee
where $(m_1, \dots, m_k|n_1, \dots,n_k)$
runs through all possible Frobenius notations.
After the boson-fermion correspondence,
\be\begin{split}
& Z_{dessins} \\
= & \sum_{(m_1, \dots, m_k|n_1, \dots,n_k)}
\det (A_{m_i, n_j})_{1 \leq i, j \leq k} \cdot
(-1)^{\sum_{j} n_j}
s_{(m_1, \dots, m_k|n_1, \dots,n_k)}.
\end{split}
\ee
Now we plug in \eqref{eqn:Amn} to get:
\ben
&& Z_{dessins} \\
&= & \sum_{(m_1, \dots, m_k|n_1, \dots,n_k)}
\det \biggl(
\frac{(-1)^{n_j} s^{m_i+n_j+1}uv}{(m_i+n_j+1)m_i!n_j!}
\prod_{a=1}^{m_i} (u+a)(v+a) \\
&& \qquad \qquad \qquad \qquad \cdot \prod_{b=1}^{n_j} (u-b)(v-b)
\biggr)_{1 \leq i, j \leq k} \\
&& \cdot
(-1)^{\sum_{j} n_j}
s_{(m_1, \dots, m_k|n_1, \dots,n_k)} \\
&= & \sum_{(m_1, \dots, m_k|n_1, \dots,n_k)}
\det \biggl(
\frac{1}{(m_i+n_j+1)m_i!n_j!}
\biggr)_{1 \leq i, j \leq k} \\
&& \cdot
s^{\sum_i m_i+\sum_jn_j+k}uv \prod_{a=1}^{m_i} (u+a)(v+a) \cdot \prod_{b=1}^{n_j} (u-b)(v-b)
\cdot s_{(m_1, \dots, m_k|n_1, \dots,n_k)} \\
& = & \sum_{(m_1, \dots, m_k|n_1, \dots,n_k)}
\det \biggl(
\frac{1}{(m_i+n_j+1)m_i!n_j!} \biggr)_{1 \leq i, j \leq k} \\
&& \cdot s^{|\lambda|} \prod_{x\in \lambda} (u+c(x))(v+c(x))
\cdot s_\lambda,
\een
where $\lambda$ is the partition function
corresponding to $(m_1, \dots, m_k|n_1, \dots, n_k)$.
To finish the proof,
it suffices to show that
\be \label{eqn:Hook}
\det \biggl(
\frac{1}{(m_i+n_j+1)m_i!n_j!} \biggr)_{1 \leq i, j \leq k}
= \prod_{x\in \lambda} \frac{1}{h(x)}.
\ee
This can be proved as follows.
By \cite[Example 4, p. 45]{Macdonald},
\ben
s_\lambda|_{p_n=X}
= \prod_{x\in \lambda} \frac{X+c(x)}{h(x)}.
\een
In particular,
\ben
s_{(m|n)}|_{p_n=X}
= \frac{X\prod_{a=1}^m(X+a) \cdot \prod_{b=1}^n (X-b)}{(m+n+1)m!n!} .
\een
On the other hand,
by \cite[Example 9, p. 47]{Macdonald},
\ben
s_{(m_1, \dots, m_k|n_1, \dots, n_k)}
 = \det ( s_{(m_i|n_j)} )_{1 \leq i, j \leq k}.
\een
So after taking the specialization $p_n = X$ for all $n \geq 1$,
we get:
\ben
&& \prod_{x\in \lambda} \frac{X+c(x)}{h(x)}
= \det \biggl(
\frac{X\prod_{a=1}^{m_i} (X+a) \cdot \prod_{b=1}^{n_j} (X-b)}
{(m_i+n_j+1)m_i!n_j!}
\biggr)_{1 \leq i, j \leq k}.
\een
Then \eqref{eqn:Hook} is proved by
comparing the coefficients of $X^{|\mu|}$ on both sides.
This completes the proof of Theorem \ref{thm:Main1}.

\end{appendix}

\end{document}